\documentclass[aps,pra,reprint,superscriptaddress,nofootinbib,twocolumn]{revtex4}
\usepackage{amsmath}
\usepackage{amsfonts}
\usepackage{amsthm}
\usepackage{amssymb}
\usepackage{graphicx}
\usepackage{enumerate}
\usepackage{color}
\usepackage[T1]{fontenc}
\usepackage{mathptmx}
\usepackage{braket}
\usepackage{todonotes}
\usepackage{soul}
\usepackage{subfigure}
\usepackage{tabularx}

\usepackage{mathrsfs}
\usepackage{float}
\usepackage{amsmath}
\usepackage{bm}
\usepackage{multirow}
\usepackage{extarrows}

\graphicspath{{figures/}}

\DeclareMathOperator{\Tr}{Tr}

\newtheorem{theorem}{Theorem}
\newtheorem{lemma}{Lemma}

\newtheorem{definition}{Definition}

\newtheorem*{theorem*}{Theorem}
\newtheorem*{lemma*}{Lemma}
\newtheorem*{proposition*}{Proposition}

\newcounter{protocol}

\def\>{\rangle}
\def\<{\langle}
\def\zo{$\{0,1\}$}

\usepackage[bookmarks=false,colorlinks,citecolor=blue,linkcolor=blue,anchorcolor=blue,urlcolor=blue]{hyperref}
\begin{document}
	
	\title{Optimal and Feasible Contextuality-based Randomness Generation}
	
	\author{Yuan Liu}
	\email{yuan59@connect.hku.hk}
	\affiliation{School of Computing and Data Science, The University of Hong Kong, Pokfulam Road, Hong Kong}
	\author{Ravishankar Ramanathan}
	\email{ravi@cs.hku.hk}
	\affiliation{School of Computing and Data Science, The University of Hong Kong, Pokfulam Road, Hong Kong}
	
	\begin{abstract}
Semi-device-independent (SDI) randomness generation protocols based on Kochen-Specker contextuality offer the attractive features of compact devices, high rates, and ease of experimental implementation over fully device-independent (DI) protocols. Here, we investigate this paradigm and derive four results to improve the state-of-art. Firstly, we introduce a family of simple, experimentally feasible orthogonality graphs (measurement compatibility structures) for which the maximum violation of the corresponding non-contextuality inequalities allows to certify the maximum amount of $\log_2 d$ bits of randomness from a qu$d$it system with projective measurements for $d \geq 3$. We analytically derive the Lovász theta and fractional packing number for this graph family, and thereby prove their utility for optimal randomness generation in both randomness expansion and amplification tasks. Secondly, a central additional assumption in contextuality-based protocols over fully DI ones, is that the measurements are repeatable and satisfy an intended compatibility structure. We frame a relaxation of this condition in terms of $\epsilon$-orthogonality graphs for a parameter $\epsilon > 0$, and derive quantum correlations that allow to certify randomness for arbitrary relaxation $\epsilon \in [0,1)$. Thirdly, it is well known that a single qubit is non-contextual, i.e., the qubit correlations can be explained by a non-contextual hidden variable (NCHV) model. We show however that a single qubit is \textit{almost} contextual, in that there exist qubit correlations that cannot be explained by $\epsilon$-faithful NCHV models for small $\epsilon > 0$. Finally, we point out possible attacks by quantum and general consistent (non-signalling) adversaries for certain classes of contextuality tests over and above those considered in DI scenarios. 
	
		
	\end{abstract}
	\maketitle

	\textit{Introduction.--}	
Contextuality is a defining feature of quantum theory, distinguishing it from classical (non-contextual) hidden variable theories. The phenomenon of contextuality (specifically Kochen-Specker contextuality aka outcome contextuality) shows that outcomes cannot be assigned to quantum measurements (in Hilbert spaces of dimension $d \geq 3$) independently of the particular contexts in which the measurements are realized. The fact that the measurements of quantum observables may not be thought of as revealing pre-determined properties implies a fundamental or intrinsic randomness that can be exploited in quantum randomness expansion and amplification protocols~\cite{ekert1991quantum,herrero2017quantum}. Such protocols offer the promise of high-quality quantum-certified secure random bits that would be crucial for several cryptographic applications.  
	
While fully device-independent (DI) protocols based on tests of quantum nonlocality have been developed in recent years for the tasks of randomness generation, expansion and amplification~\cite{pironio2010random,pironio2013security,brown2019framework,liu2021device,colbeck2012free,gallego2013full,ramanathan2016randomness,brandao2016realistic,ramanathan2018practical,kessler2020device,zhao2022tilted,ramanathan2021no,ramanathan2023finite}, their experimental implementation is highly demanding, with the requirement of high-fidelity entanglement distributed between multiple devices and high detection efficiencies. In contrast, contextuality can be tested in a single quantum system, which significantly simplifies the experimental requirements. The price to pay however is the additional assumption that the measurements obey specified compatibility relations and that the same observable is measured in multiple contexts. 
As such, contextuality-based protocols are said to be semi-device-independent (SDI), note however that other than these basic assumptions, no other assumption is required on the state, measurements or dimensionality of the system. The SDI framework thus offers a practical compromise between security and feasibility, allowing for compact, experimentally feasible high-rate protocols. 

\color{black}
While multiple investigations of contextuality-based SDI protocols have been carried out in recent years \cite{abbott2012strong,abbott2014value,abbott2015variant,um2020randomness,um2013experimental,singh2017quantum,singh2024local,ramanathan2020gadget,liu2023optimal}, several fundamental open questions remain. 
One of the most crucial aspects of generation protocols is the randomness generation rate. It is well known that the maximum theoretical amount of randomness that can be certified from a $d$-dimensional quantum system using projective measurements is $\log_2 d$ bits. However, it was previously unknown how to achieve this fundamental limit for arbitrary $d \geq 3$~\cite{acin2016optimal,borkala2022device,sarkar2023self}. Only a recent work by M.~Farkas et al.~\cite{farkas2024maximal} achieved the maximum randomness for any $d \geq 2$ using POVMs in Bell scenarios, but this necessitates consuming entangled qudits.
In this Letter, we introduce a family of measurement structures (represented as orthogonality graphs) and a family of non-contextuality inequalities parametrized by dimension $d \geq 3$ to achieve this goal. We derive the maximal values of these inequalities in non-contextual hidden variable (NCHV) theories, quantum theory, and general consistent (or non-disturbing \cite{RSKD12}) theories. We prove that observing the maximum quantum violation certifies $\log_2 d$ bits of randomness, even when the measurements are chosen using arbitrarily weak Santha-Vazirani-type source of randomness \cite{colbeck2012free, gallego2013full, ramanathan2016randomness}.

A major practical challenge in contextuality-based protocols is that one cannot precisely distinguish between a projector and an arbitrarily close POVM element in experiments~\cite{kunjwal2018statistical,xu2016reformulating}. To address this, Winter~\cite{winter2014does} introduced the concept of $\epsilon$-ontologically faithful NCHV ($\epsilon$-ONC) models, stating that experiments can only refute NCHV models that are $\epsilon$-faithful, meaning that POVM elements that are $\epsilon$-close can only be assigned different values in the NCHV model with probability $\epsilon$. A natural follow-up question arises: To what extent can quantum theory refute such $\epsilon$-ONC models? In other words, {\em can quantum theory refute such models for arbitrary $\epsilon$?} Furthermore, a single qubit is known to be non-contextual, as its correlations can always be explained by NCHV models~\cite{kochen1990problem,bell1966problem,mermin1993hidden}. The fact that only $\epsilon$-faithful NCHV models are meaningful and can be refuted in realistic experiments  raises the question: {\em Could this inevitable $\epsilon$-faithfulness make a single qubit system ``almost'' contextual, and thus useful for certain contextuality-based tasks?} In this Letter, we answer both questions in the affirmative. We revisit the notion of $\epsilon$-ONC models, derive a tight bound on the value achieved by $\epsilon$-ONC models for general non-contextuality inequalities, and use it to show that quantum theory refutes such models for arbitrary $\epsilon \in [0,1)$. Additionally, we prove that a single qubit is almost contextual - there exist qubit correlations that cannot be explained by $\epsilon$-ONC models for small $\epsilon > 0$. Finally, to highlight the significance of the contextuality-based SDI protocol proposed in this Letter, we investigate the security of these protocols by deriving a class of attacks that can be carried out by adversaries holding quantum or more general consistent resources~\cite{bierhorst2018experimentally}. These attacks extend those examined in the fully DI setting, providing a comprehensive analysis of potential vulnerabilities in contextuality-based SDI protocols.

\color{black}

\textit{Certifying $\log_2 d$ bits of randomness from qudit systems for $d \geq 3$.-- }
	In a randomness generation protocol, the optimal amount of randomness that can be certified from a qu$d$it system using projective measurements is $\log_2 d$. We introduce a family of non-contextuality inequalities for each dimension $d \geq 3$ to achieve this optimal value. More precisely, we show that when the maximum quantum values of these non-contextuality expressions are achieved, one of the measurement setttings necessarily yields $d$ outcomes each occurring with uniform probability $1/d$, with respect to any quantum adversary (holding general quantum side information). 
	
The corresponding measurement structures are illustrated in Fig.~\ref{ORgraph} as orthogonality graphs. Specifically, the measurement structure in dimension $d$ is defined by a finite set consisting of $5d$ rank-one projectors $S_d=\{|v_{i,j}\>\<v_{i,j}|\}_{i=1,\ldots,d;j=1,\ldots,5}$. The projectors exhibit the following orthogonality structure: (1) projectors in $\{|v_{i,1}\>\<v_{i,1}|\}_{i=1}^d$ are mutually orthogonal, and (2) for each $i \in \{1,\ldots,d\}$, the projector $|v_{i,j}\>\<v_{i,j}|$ is orthogonal to $|v_{i,j+1}\>\<v_{i,j+1}|$ with $j+1$ taken modulo $5$. The corresponding orthogonality graph $\mathcal{G}_d=(\mathcal{V}_d,\mathcal{E}_d)$ is drawn with each vertex in  $\mathcal{V}_d$ representing a projector in $S_d$ and two vertices connected by an edge if and only if the corresponding projectors are orthogonal. From the orthogonality relations among the projectors in $S_d$, we see that the orthogonality graph $\mathcal{G}_d$ contains a central maximum clique of size $d$ along with $d$ cycles of length $5$ ($C_5$), with every vertex in the maximum clique belonging to a single $C_5$.

\begin{figure}[t]
		\centering
		\subfigure[]{
			\begin{minipage}[t]{0.23\textwidth}
				\centering
				\includegraphics[width=1\textwidth]{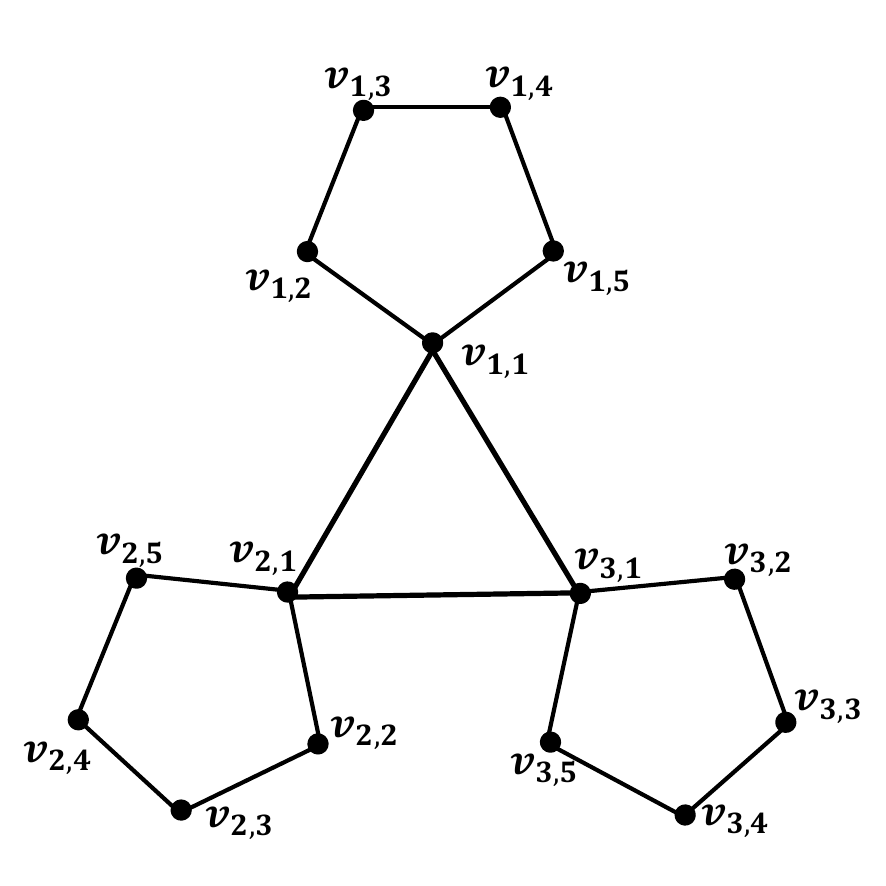}
			\end{minipage}
		}
		\subfigure[]{
			\begin{minipage}[t]{0.23\textwidth}
				\centering
				\includegraphics[width=1\textwidth]{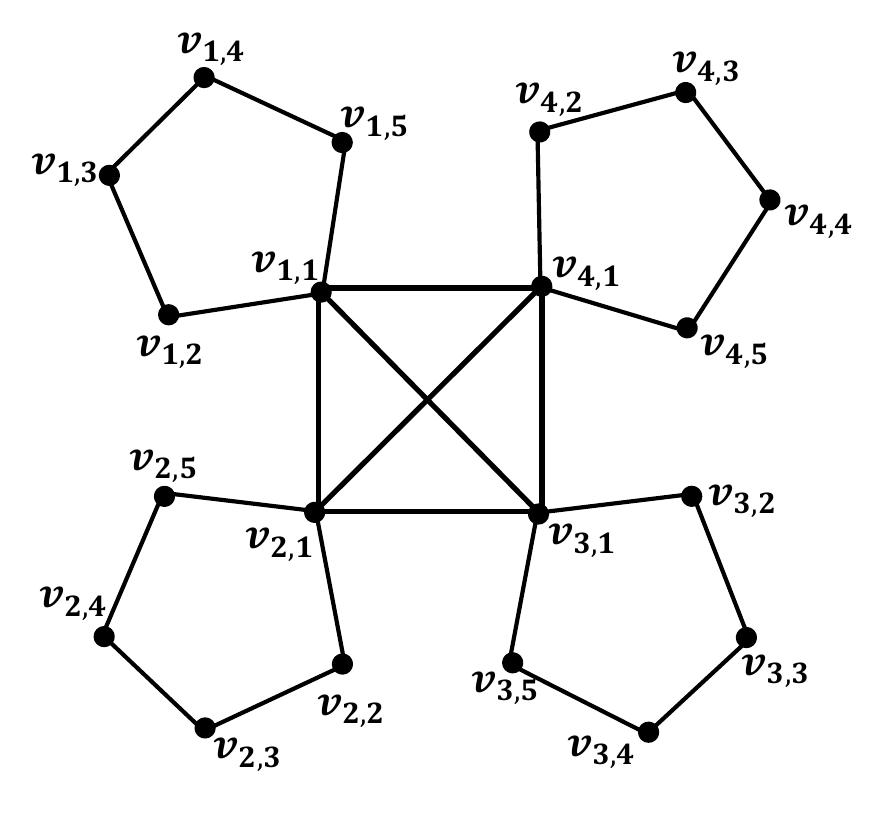}
			\end{minipage}
		}
		\caption{The orthogonality graphs $\mathcal{G}_3$, $\mathcal{G}_4$ for rank-one projector sets $S_3$ and $S_4$ described in the main text.}
		\label{ORgraph}
\end{figure}
	
	From the set of projectors, one can also construct a set of binary observables defined by $A_{i,j}:= \mathbb{I}-2|v_{i,j}\>\<v_{i,j}|$, with eigenvalues $\{+1, -1\}$, whose compatibility structure follows directly from the orthogonality relations of the projectors. Compatible observables can be jointly measured, and their measurement outcomes in quantum theory are unaffected by the order in which they are measured.
We now present the following non-contextuality expression $\mathcal{I}_d$ for any $d \geq 3$, ensuring that each term (referring to one context) in $\mathcal{I}_d$ consists of compatible observables:
	\begin{equation}
		\mathcal{I}_d := 2\prod_{i=1}^d A_{i,1} + \sum_{i=1}^{d} \sum_{j=1}^{5} A_{i,j} A_{i,j+1}.
	\end{equation}
    
The NCHV bound and quantum bound of the above non-contextuality expression are directly related to the weighted independence number and Lovász theta number of the orthogonality graph~\cite{cabello2014graph}. For a brief explanation of these graph parameters in relation to quantum contextuality, please refer to App.~A in the Supplemental Material~\cite{SM}. To be precise, we rewrite the expression $\mathcal{I}_d$ using the orthogonality relations of the projectors in $S_d$:
	\begin{equation}
		\begin{split}
			\mathcal{I}_d &= 2\mathbb{I} - 4\sum_{i=1}^d |v_{i,1}\>\<v_{i,1}| + \sum_{i=1}^d \left(5\mathbb{I} - 4\sum_{j=1}^5 |v_{i,j}\>\<v_{i,j}|\right) \\
			&= (5d + 2)\mathbb{I} - 4\sum_{i=1}^d \left(2 |v_{i,1}\>\<v_{i,1}| + \sum_{j=2}^5 |v_{i,j}\>\<v_{i,j}|\right).
		\end{split}
	\end{equation}
	The NCHV lower bound for the expectation value of $\mathcal{I}_d$ is calculated based on all possible \zo-valued  non-contextual assignments for projectors in $S_d$, which is equivalent to the weighted independence number $\alpha$ of the orthogonality graph $\mathcal{G}_d$, with weight $2$ assigned to vertices $v_{i,1} \ \forall i$ and weight $1$ to the other vertices. Denote this assignment of weights to the vertices of the graph by $\bm{w}$. Quantum theory allows to achieve a value lower than the NCHV bound, characterized by the weighted Lovász theta number $\theta$ of the orthogonality graph $\mathcal{G}_d$ (given by a semidefinite program). In general consistent (no-disturbance) theories, the value of the non-contextuality expression is characterized in terms of the weighted fractional packing number $\alpha^*(\mathcal{G}_d,\bm{w})$ (given by a linear program).
	\begin{equation}
		\<\mathcal{I}_d\> \overset{\text{NCHV}}{\geq} (5d + 2) - 4\alpha(\mathcal{G}_d,\bm{w}) \overset{\text{Q}}{\geq} (5d + 2) - 4\theta(\mathcal{G}_d,\bm{w}).
	\end{equation}
			\begin{table}[H]
			\centering
			\footnotesize
			\begin{tabular}{|c||c|c|c|c|c|c|c|}
				\hline
				$d$ & 3 & 4 & 5 & 6 & 7 & 8 & $d\geq 9$ \\
				\hline
				\hline
				$\alpha^*(\mathcal{G}_d,\bm{w})$ & 8 & 10 & 12 & 14 & 16 & 18 & $2d+2$ \\
				\hline
				$\theta(\mathcal{G}_d,\bm{w})$ & 7.6753 & 9.8030 & 11.8869 & 13.9419 &  15.9762& 17.9944 & $2d+2$ \\
				\hline
				$\alpha(\mathcal{G}_d,\bm{w})$ & 7 & 9 & 11 & 13 & 15 & 17 & $2d+1$ \\
				\hline
			\end{tabular}
			\caption{The weighted fractional packing number, weighted independence number and weighted Lovász theta number of the orthogonality graphs $(\mathcal{G}_d,\bm{w})$.}
			\label{table1}
		\end{table}
        The NCHV, quantum, and no-disturbance lower bounds of the non-contextual expression $\mathcal{I}_d$ are given by the values in Table~\ref{table1}. Furthermore, we demonstrate that any quantum correlation achieving $\theta(\mathcal{G}_d,\bm{w})$ necessarily results in $d$ outcomes with uniform probabilities of $1/d$ for the central clique, for any $d \geq 3$. In particular, we establish the following (proof in App.~B of~\cite{SM}).

		\begin{theorem}\label{them_QuantumValue}
		 The weighted independence number of the graph $(\mathcal{G}_d,\bm{w})$ is $\alpha(\mathcal{G}_d,\bm{w}) = 2d + 1$ for any $d \geq 3$, where vertices $v_{i,1} \ \forall i$ are assigned weight $2$ and all other vertices are assigned weight $1$. The weighted Lovász theta number for these weights is $\theta(\mathcal{G}_d,\bm{w}) = 2d+2$ for $d \geq 9$, with the value for $3 \leq d \leq 8$ being given as in the Table~\ref{table1}. The weighted fractional packing number is $\alpha^*(\mathcal{G}_d,\bm{w}) = 2d+2$ for $d \geq 3$. Furthermore, for any quantum correlation achieving the weighted Lovász theta number, there exists a measurement setting (a context) $c^*$ that produces $d$ outcomes, each occurring with a uniform probability of $1/d$.
	\end{theorem}

\textit{Use in randomness expansion and amplification.-- }
\color{black}
The aforementioned contextuality tests that certify $\log_2 d$ bits of randomness can be utilized in contextuality-based SDI randomness expansion protocols, which aim to transform a short random seed into a longer string of uniformly random bits. In App.~F of~\cite{SM}, we provide a detailed description of how to adapt the standard randomness expansion protocol~\cite{brown2019framework} into a contextuality-based one. Typically, the protocol operates over multiple rounds, designating each round as either a {\em test round} or a {\em randomness generation round}. In test rounds, an initial random seed is used to select the input for the contextuality test, and the data collected from these rounds are used to calculate the {\em game score} $\omega_d := \frac{(5d+2) - \<\mathcal{I}_d\>}{4}\in [\alpha(\mathcal{G}_d,\bm{w}),\theta(\mathcal{G}_d,\bm{w})]$. In contrast, randomness generation rounds use a specific context $c^*$ as the input and accumulate the outcomes as raw randomness. The amount of randomness can be formally quantified by the min-entropy (denoted as $H_{\text{min}}(a | X = c^*, E)$) of the guessing probability of a potential quantum adversary, which represents the highest probability that an adversary can correctly guess the outcomes from $c^*$. Importantly, in the SDI (and also the DI) paradigm, the quality of the generated randomness is assessed based on the game score calculated from the test rounds. For the case of $d = 3$, we numerically calculate the randomness quantified by $H_{\text{min}}(a | X = c^*, E)$ as a function of the game score $\omega_3$ and plot it in Fig.~\ref{min_entropy}. As shown in the figure, the observation of the maximum value $\theta(\mathcal{G}_3,\bm{w})$ certifies $\log_2 3$ bits of randomness.

The Entropy Accumulation Theorem (EAT)~\cite{arnon2018practical, dupuis2020entropy} provides a lower bound on the total output randomness for the protocol through i.i.d. analyses. We review the definition of EAT channels in App.~F, adapting them to the contextuality scenario. To successfully apply the EAT theorem, it is essential to identify the so-called min-tradeoff functions $f_{\text{min}}$, which map the game score to $\mathbb{R}$ and provide lower bounds for the worst-case von Neumann entropy. Following the approach outlined in~\cite{brown2019framework}, we derived a family of convex, differentiable functions $\{f_{\hat{\omega}}(\omega_3)\}_{\hat{\omega}}$ as candidates for the min-tradeoff function for the protocol, which are plotted in Fig.~\ref{min_entropy} (details are deferred to App.~F). Each of these functions is a valid min-tradeoff function for the protocol and can be applied as $f_{\text{min}}$ to lower bound the total randomness; however, as shown in Fig.~\ref{min_entropy}, they exhibit different performance for the exact game score obtained from the test rounds.

    \color{black}
	
	\begin{figure}[t]
		\centering
		\includegraphics[width=0.7\linewidth]{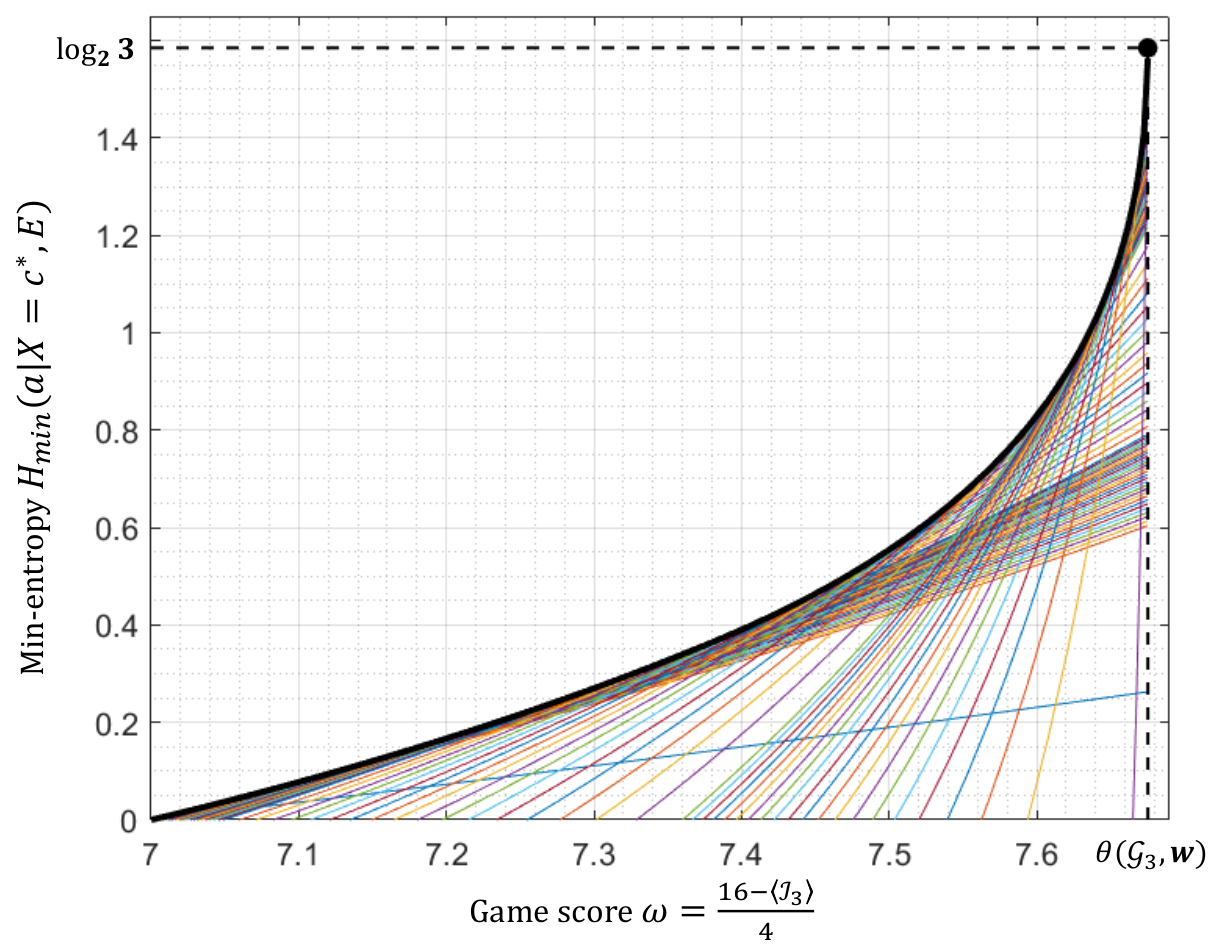}
		\caption{The min-entropy $H_{\text{min}}(a | X = c^*, E)$ versus the game score $\omega_d$ for $d=3$, along with a family of candidate min-tradeoff functions for the protocol.}
		\label{min_entropy}
	\end{figure}

Furthermore, we note that the derived non-contextuality inequalities serve as ideal candidates for randomness amplification protocols \cite{colbeck2012free,gallego2013full,ramanathan2016randomness,brandao2016realistic,ramanathan2018practical,kessler2020device,zhao2022tilted,ramanathan2021no,ramanathan2023finite}. In this class of protocols, one starts with a weak seed, such as a Santha-Vazirani source, where each bit has a small amount of randomness defined by a bias $0 \leq \kappa < 1/2$, conditioned on any adversarial side information. Specifically, we have $1/2 - \kappa \leq \text{Pr}[X_j = 0 | X_0, X_1, \ldots, X_{j-1}] \leq 1/2 + \kappa$ for all $j$, where $X_j$ denote the bits from the source. The task of randomness amplification is to extract private, fully random bits from such a weak random source. It is well-known that to achieve randomness amplification from sources with arbitrary $\kappa < 1/2$, one requires quantum correlations that achieve (or come arbitrarily close to) the general consistent (non-signalling) bound for the test~\cite{gallego2013full,liu2024equivalence}. The fact that the derived inequalities certify the maximum amount of $\log_2 d$ bits for $d \geq 9$, under the condition that $\alpha(\mathcal{G}_d,\bm{w}) < \theta(\mathcal{G}_d,\bm{w}) = \alpha^*(\mathcal{G}_d,\bm{w})$ (see Table \ref{table1}) makes them ideal candidates for this task.

\color{black}

In this regard, it is important to note that not all contextuality tests are useful for randomness expansion and amplification. Specifically, we propose two types of attacks in App.~G of~\cite{SM}, which highlight the significance of our protocol. Firstly, we point out that protocols that rely solely on the observed value of a \textit{state-independent} contextuality inequality are insecure against a quantum adversary who shares entanglement with the device implementing the protocol. Secondly, we show that protocols based on a specific class of contextuality tests referred to as {\em magic arrangements}~\cite{arkhipov2012extending}, which include the famous Peres-Mermin magic square and the Mermin magic star, are insecure against hybrid classical adversaries~\cite{bierhorst2018experimentally}. These adversaries do not share any quantum or post-quantum resources with the device but are allowed to prepare general consistent correlations in advance for the device.
    
\color{black}	
\textit{Relaxed $\epsilon$-Non-Contextuality for Experimental Feasibility.-- } 
The central additional assumption in our SDI protocols, as opposed to fully DI protocols, is non-contextuality, which indicates that the ``same'' projector is measured in multiple different contexts. That is, an outcome $i$ from one measurement is identified with an outcome $j$ from another measurement, even if they arise from completely different experiments. In the NCHV model, the random variables $f_i, \; f_j \in \{0,1\}$ associated with these outcomes are identical. Similarly, in the quantum model, the associated projectors (or more generally, POVM elements) are identical. Nevertheless, in practical contextuality experiments (and also the resulting SDI protocols), this assumption is not feasible since experiments cannot perfectly distinguish between a projector $P_i$ and an arbitrarily close POVM element $Q_i$. 
The general NCHV model can assign different values in $\{0,1\}$ to these close POVM elements, leading to a ``nullification'' of the Bell-Kochen-Specker theorem. Meyer, Clifton, and Kent \cite{meyer1999finite,kent1999noncontextual,clifton2000simulating} provided such a nullification by showing that for each $d \geq 3$, there exists a dense set of complete projective measurements made up of rank-one projectors, such that each projector appears in only one measurement.
Consequently, in practical experimental situations and SDI protocols, we do not certify quantum contextuality against arbitrary NCHV models. Instead, we must consider $\epsilon$-ontologically faithful non-contextual models for a parameter $\epsilon \in [0,1)$, a concept formally introduced by Winter in \cite{winter2014does}. Here, we adapt the definition for orthogonality graphs, establish a tighter upper bound for non-contextuality inequalities than previously noted in \cite{winter2014does}, and demonstrate that quantum contextuality can be certified for arbitrary $\epsilon$.

\color{black}
\begin{definition}
An $\epsilon$-ontologically faithful non-contextual ($\epsilon$-ONC) model for an orthogonality graph $G = (V, E)$, where we denote the set of maximal cliques of $G$ as $\Gamma$, consists of measurable functions $f_v^{C} : \Lambda \to \{0,1\}$ for vertices $v \in V$ and contexts $C \in \Gamma$ such that
\begin{eqnarray}\label{onc_condition}
&&\sum_{v \in C} f_v^{C} \leq 1, \quad \forall C \in \Gamma, \nonumber \\	
&&\text{Pr}\left[f_v^{C} \neq f_{v}^{C'}  \right] \leq \epsilon, \quad \forall C, C' \in \Gamma, \; \forall v \in C \cap C'.
\end{eqnarray}		 
Here, $\Lambda$ is a joint probability space with probability measure $\text{Pr}$ from which the functions $f_v^{C}$ are derived.
\end{definition}

In this relaxed NCHV model, we treat the observable corresponding to vertex $v$ in contexts $C$ and $C'$ as ``almost'' identical; specifically, the probability that the random variables $f_v^{C}$ and $f_{v}^{C'}$ differ is at most $\epsilon$. To incorporate this relaxed model with a graph-theoretic representation, we introduce the concept of $\epsilon$-orthogonality, denoted by $\epsilon$-edges in an $\epsilon$-orthogonality graph. In contrast to edges in an orthogonality graph, which indicate strict orthogonality, $\epsilon$-edges represent ``almost'' orthogonality.

	\begin{definition}
		An $\epsilon$-orthogonality graph $G_{\epsilon}$ for a set $V$ is defined as a triple $G_{\epsilon} = (V_{\epsilon}, E_{\epsilon}, \widetilde{E}_{\epsilon})$, where $V_{\epsilon}$ is the set of vertices, $E_{\epsilon}$ is the edge set representing strict orthogonality, and $\widetilde{E}_{\epsilon}$ is the $\epsilon$-edge set representing $\epsilon$-orthogonality.
	\end{definition}

We now apply the $\epsilon$-ONC model to a non-contextuality inequality and calculate its maximum value in the relaxed model using the graph-theoretic approach. For a given non-contextuality inequality and its corresponding orthogonality graph $G = (V,E)$ with vertex weights $\bm{w}$, we construct an $\epsilon$-orthogonality graph $G_{\epsilon} = (V_{\epsilon}, E_{\epsilon}, \widetilde{E}_{\epsilon})$ with vertex weights $\bm{w}_{\epsilon}$. Denote the set of maximal cliques of $G$ as $\Gamma$, and let $n_v$ be the number of maximal cliques in $\Gamma$ containing $v\in V$. We define $V_{\epsilon} := \{v^{(C)} \mid v \in V, C \in \Gamma\}$. For each $(v,u)\in E$ and $C,C'\in \Gamma$ with $v\in C$ and $u\in C'$, $G_{\epsilon}$ has an edge between $v^{(C)}$ and $u^{(C')}$: $(v^{(C)},u^{(C')})\in E_{\epsilon}$ when $C=C'$, and $(v^{(C)},u^{(C')})\in \widetilde{E}_{\epsilon}$ when $C\neq C'$. The vertex weights $\bm{w}_{\epsilon}$ for $v^{(C)}$ are defined as $w_{\epsilon}(v^{(C)}) := w(v)/n_v$, where $w(v)$ is the weight of $v$ in $\bm{w}$.

\color{black}
In a graph two vertices connected by an edge have zero ``independence'', they cannot belong to the same independent set. On the other hand, non-adjacent vertices have complete ``independence''. $\epsilon$-edges offer an intermediate level of independence between two vertices. This translates into the idea that the probability that two purportedly orthogonal projectors are both assigned value $1$ is at most $\epsilon$. 

\color{black}

To calculate the maximum value of the given non-contextuality inequality in the $\epsilon$-ONC model, denote as $\alpha_{\epsilon}(G_{\epsilon}, \bm{w}_{\epsilon})$, we define two related simple graphs: $G' := (V_{\epsilon}, E_{\epsilon})$, which includes only the edges of $G_{\epsilon}$ (and is thus just a bunch of disjoint maximal cliques), and $G'' := (V_{\epsilon}, E_{\epsilon} \cup \widetilde{E}_{\epsilon})$, which treats the $\epsilon$-edges as ``standard'' edges. We prove in App.~C in~\cite{SM} that:
	\begin{equation}\label{ep_alpha_1}
		\alpha_{\epsilon}(G_{\epsilon},\bm{w}_{\epsilon}) \leq \epsilon \cdot \alpha(G',\bm{w}_{\epsilon}) + (1 - \epsilon) \cdot \alpha(G'',\bm{w}_{\epsilon}).
	\end{equation}
    
Having established the value of non-contextuality inequalities under the experimentally testable $\epsilon$-ONC models, we now show that there exist quantum correlations that exceed this value for certain inequalities for arbitrary value of $\epsilon \in [0,1)$. 
	\begin{theorem}\label{them_relax}
        For an arbitrary $\epsilon \in [0, 1)$, there exists a non-contextuality inequality and a corresponding optimal quantum correlation that certifies contextuality against $\epsilon$-ONC models.
	\end{theorem}
    \color{black}
	We prove this theorem (details are in App.~C of~\cite{SM}) by showing that the threshold $\epsilon$ for orthogonality graphs correspond to odd cycles $C_n$ is given by $\epsilon \leq 1 - \frac{\pi^2}{4n}$.
	Therefore, for any $\epsilon \in [0, 1)$, the optimal quantum correlations for the odd cycle $C_n$ with odd $n \geq \left\lceil \frac{\pi^2}{4(1 - \epsilon)} \right\rceil$, serve to certify contextuality.

    \textit{$\epsilon$-Contextuality for a single qubit system.-- }	
    It is well known that a single qubit is non-contextual, as the KS theorem applies only to Hilbert spaces of dimension $\geq 3$. 
    Indeed, explicit NCHV models have been derived to explain the correlations for a single qubit~\cite{kochen1990problem,bell1966problem,mermin1993hidden}. However, these models do not satisfy the $\epsilon$-ONC conditions, as the probability, that hidden variables assign to a ``almost''-orthogonal projectors pair the same values, is significantly greater than their overlap (see App.~E of~\cite{SM}).
    Although some arguments for qubit contextuality exist~\cite{cabello2003kochen}, they rely on POVMs rather than projectors~\cite{grudka2008there}. Nevertheless, as we discussed above, since experimental errors and imperfections are unavoidable, only the rejection of $\epsilon$-ONC models is meaningful in realistic situations. This suggests us to consider projectors under such near-orthogonalities when calculating the theoretical quantum bounds of non-contextuality inequalities. In this section, we find that this relaxation enables the possibility of qubit contextuality.

    To explore this further, we define the $\epsilon$-Lovász theta number associated with orthogonal representations of $G_{\epsilon}$ (with vertex weights $\bm{w}_{\epsilon}$), allowing projectors corresponding to vertices connected by $\epsilon$-edges to overlap by $\epsilon$. 
    \begin{equation}
		\begin{split}
			\theta_{\epsilon}(G_{\epsilon}, \bm{w}_{\epsilon}) := \max_{|\psi\>, \{|v\>\}} & \quad \sum_{v \in V} w_{\epsilon}(v) |\<\psi | v\>|^2 \\
			\text{subject to:} & \quad \||\psi\>\|^2 = 1, \||v\>\|^2 = 1, \forall v \in V, \\
			& \quad \<v | u\> = 0, \forall (v, u) \in E_{\epsilon}, \\
			& \quad |\<v | u\>| \leq \sqrt{\epsilon}, \forall (v, u) \in \widetilde{E}_{\epsilon}.
		\end{split}
		\label{ep_theta}
	\end{equation}
    This parameter is effectively bounded by the solution of an SDP relaxation of the discrete optimization problem for the $\epsilon$-independence number~Eq.~\eqref{ep_alpha_1} (proof provided in App.~D of~\cite{SM}).

	\begin{lemma}\label{lemma_theta}
		For any $\epsilon$-graph $G_{\epsilon}$ with vertex weights $\bm{w}_{\epsilon}$, the $\epsilon$-Lovász theta number is lower-bounded by the SDP relaxation of the $\epsilon$-independence number, i.e., $\theta_{\epsilon}(G_{\epsilon}, \bm{w}_{\epsilon}) \geq \theta'_{\epsilon}(G_{\epsilon}, \bm{w}_{\epsilon})$, where
		\begin{equation}
					\label{sdp_ep_theta}
			\begin{split}
				\theta'_{\epsilon}(G_{\epsilon}, \bm{w}_{\epsilon}) := \max_{X, Y} & \; \sqrt{\epsilon}(\sqrt{\bm{w}_{\epsilon}}^T X \sqrt{\bm{w}_{\epsilon}}) + (1 - \sqrt{\epsilon})  (\sqrt{\bm{w}_{\epsilon}}^T Y \sqrt{\bm{w}_{\epsilon}}) \\
				\text{s.t.} & \quad X_{vu} = 0, \; \; \forall (v, u) \in E_{\epsilon}, \\
				& \quad Y_{vu} = 0, \; \; \forall (v, u) \in E_{\epsilon} \cup \widetilde{E}_{\epsilon}, \\
				& \quad X_{vv} = Y_{vv}, \; \; \forall v \in V_{\epsilon}, \\
				& \quad \Tr(X) = 1, \; \; \Tr(Y) = 1, \\
				& \quad X \succeq 0, \;\; Y \succeq 0. 		
			\end{split}
		\end{equation}
	\end{lemma}


	

Interestingly, the above $\epsilon$-orthogonality relations enable non-trivial $\epsilon$-orthogonality graph structures in qubit systems. We demonstrate that for any $\epsilon \in (0,\frac{1}{2})$, qubit systems with projective measurements exhibit $\epsilon$-contextuality (details can be found in App.~E of~\cite{SM}). Specifically, for these values of $\epsilon$, one can find graphs $G_{\epsilon}$ and vertex weights $\bm{w}_{\epsilon}$ such that $\theta_{\epsilon}(G_{\epsilon}, \bm{w}_{\epsilon}) > \alpha_{\epsilon}(G_{\epsilon}, \bm{w}_{\epsilon})$.

\begin{theorem}\label{them_qubit}
		There exist non-contextuality inequalities and corresponding qubit correlations that certify quantum contextuality against $\epsilon$-ONC models for $\epsilon \in (0, \frac{1}{2})$, thereby exhibiting $\epsilon$-contextuality for a qubit.
\end{theorem}
	
	\textit{Conclusions and Open Questions.-- } In this Letter, we have derived a family of contextuality tests to certify the maximum possible amount of $\log_2 d$ bits of randomness from $d$-dimensional quantum systems. The tests are experimentally feasible and incorporate practical constraints including measurement settings being chosen with weak random seeds and a relaxation of the assumption that the measurements exactly conform to a prescribed compatibility structure. We have also seen that in this practical situation, even a single qubit is $\textit{almost}$ contextual, thereby introducing the possibility of utilising single qubit systems in contextuality applications. A few open questions still remain. It would be interesting to implement the prescribed tests in an actual experimental protocol and identify the advantage in rates and detection efficiency over other implementations such as those based on self-testing \cite{singh2024local}. It is an open question to establish, in the tradition of KS contextuality investigations, the orthogonality graphs with fewest vertices that certify $\log_2 d$ bits of randomness for arbitrary $d \geq 3$, and identify ones that give the best tolerance to noise, weak seeds and relaxations of the classical non-contextual model. This corresponds to the optimal contextuality tests with the fewest measurement settings. For this purpose, it would be useful to identify graphs such that their Lovász theta equals the fractional packing number, for which constructive methods have been developed based on local complementation orbits \cite{cabello2013exclusivity}.  

Additional references~\cite{grotschel2012geometric,milgrom2002envelope,lovasz1979shannon,renner2004smooth,navascues2007bounding,navascues2008convergent,yu2012state} are cited in the Supplemental Material~\cite{SM}.
	
\textit{Acknowledgments.-- } We thank the anonymous reviewers for their suggestions for improving the readability and clarity of the manuscript. This work is supported by the General Research Fund (GRF) Grants No. 17211122 and No. 17307925, and the Research Impact Fund (RIF) Grant No. R7035-21.

\end{document}


\title{Supplemental Material to ``Optimal and Feasible Contextuality-based Randomness Generation''}
	
	\author{Yuan Liu}
	\email{yuan59@connect.hku.hk}
	\affiliation{School of Computing and Data Science, The University of Hong Kong, Pokfulam Road, Hong Kong}
	\author{Ravishankar Ramanathan}
	\email{ravi@cs.hku.hk}
	\affiliation{School of Computing and Data Science, The University of Hong Kong, Pokfulam Road, Hong Kong}

\maketitle

\appendix

\section{Kochen-Specker contextuality}\label{app_context}
Quantum contextuality is a foundational feature of quantum mechanics that challenges classical explanations of nature by rejecting non-contextual hidden-variable (NCHV) models. In quantum mechanics, a projective measurement $M$ consists of a set of orthogonal projectors $M = \{P_1, P_2, \ldots\}$, where 
$\sum_{i} P_i = \mathbb{I}$. Each projector $P_i$ represents an outcome $i$ of the measurement $M$, and the probability of obtaining this outcome $i$ when measuring $M$ on a quantum state $|\psi\>$ is given by $p_{\psi}(i|M) = \<\psi|P_i|\psi\>$. If distinct projective measurements $M$ and $M'$ share the same projector $P_i$, the outcome probabilities for $P_i$ when testing these two measurements on the same state $|\psi\>$ are identical and equal to $\<\psi|P_i|\psi\>$. This implies a form of non-contextuality in quantum measurements: the probabilities of measurement outcomes depend only on the projectors themselves (and the state being tested) rather than the specific measurement context to which the projectors belong.

Classical theory attempts to explain this quantum phenomenon through hidden-variable models that obey the non-contextual condition. Such Non-contextual Hidden-Variable (NCHV) models assume the existence of hidden variables $\lambda$ (with an associated distribution) that predetermine measurement outcomes, such that the probability of measuring $M$ and obtaining outcome $i$ can be expressed as
\begin{equation}
	p_{\psi}(i|M) = \sum_{\lambda\in \Lambda} q_{\psi}(\lambda) f_{\lambda}(i),
	\label{NCHV_model}
\end{equation}
where $q_{\psi}:\Lambda\rightarrow [0,1]$, with $\sum_{\lambda\in \Lambda} q_{\psi}(\lambda) = 1$, is the probability distribution of hidden variables $\lambda$ associated with the state $|\psi\>$. The response function $f_{\lambda}$ is defined over the projectors $\{P_i\}$, regardless of the measurement context to which each $P_i$ belongs, and satisfies $f_{\lambda}(i) \in \{0,1\}$ with $\sum_{i, P_i \in M} f_{\lambda}(i) = 1$.

The conflict between quantum mechanics and NCHV models was studied by Bell~\cite{bell1966problem}, Kochen and Specker~\cite{kochen1990problem}, who focused on the infeasibility of defining a response function $f_{\lambda}$ that exhibits the non-contextual property. Specifically, the Kochen-Specker (KS) theorem~\cite{kochen1990problem} states that there exists a finite set of vectors (rank-one projectors) $\mathcal{V} = \{|v_1\>, \ldots, |v_{n}\>\} \subset \mathbb{C}^d$ in a $d$-dimensional Hilbert space for $d \geq 3$ for which it is impossible to find a \zo-assignment function $f: \mathcal{V} \rightarrow \{0,1\}$ satisfying the:
\begin{enumerate}
	\item Exclusivity condition: $f(|v\>) + f(|v'\>) \leq 1$ for any two orthogonal vectors $|v\>$ and $|v'\>$; and the
	\item Completeness condition: $\sum_{|v\> \in C} f(|v\>) = 1$ for any set $C\subset \mathcal{V}$ of $d$ mutually orthogonal vectors.
\end{enumerate}
A set $\mathcal{V}$ that satisfies these properties is referred to as a KS set or a proof of the KS theorem.

Reasoning about quantum contextuality has traditionally been approached using graph-theoretic representations. For any set of vectors $\mathcal{V}$, one can define an orthogonality graph $G$, where each vector $|v_i \rangle \in \mathcal{V}$ corresponds to a vertex $v_i$ in $G$, and two vertices $v_1, v_2$ are connected by an edge if and only if $\langle v_1 | v_2 \rangle = 0$. The feasibility of the $\{0,1\}$-assignment problem for a given set of vectors can then be formulated as the problem of $\{0,1\}$-coloring its orthogonality graph. 

Denote by $V$ the vertex set of $G$. A clique is a subset of mutually adjacent vertices, and $\omega(G)$ denotes the size of the largest clique in $G$. A $\{0,1\}$-coloring of a graph $G$ is a map $f: V \rightarrow \{0,1\}$ that satisfies the following conditions:
\begin{enumerate}
	\item For every clique $C$ in $G$, it holds that $\sum_{v \in C} f(v) \leq 1$;
	\item For every maximum clique $C$ in $G$ of size $\omega(G)$, there exists exactly one vertex $v \in C$ such that $f(v) = 1$.
\end{enumerate}

For a given orthogonality graph $G$, we can label each vertex $v_i$ as an input-output event $(v_i | C_{v_i})$, where $C_{v_i}$ is a maximal clique containing vertex $v_i$. By interpreting orthogonality between vertices as exclusivity between events, we can define a non-contextual expression $\mathcal{I}_G := \sum_{v_i \in V} w(v_i) p(v_i | C_{v_i})$, where $\bm{w} > \bm{0}$ represents weights assigned to the vertices. The non-contextual expression $\mathcal{I}_G$ is analogous to a Bell expression in the quantum nonlocality scenario. Moreover, the classical (NCHV), quantum, and post-quantum bounds of $\mathcal{I}_G$ are characterized by the weighted independence number, weighted Lovász theta number, and weighted fractional packing number of $(G,\bm{w})$, respectively~\cite{cabello2014graph}.

In this manuscript, each rank-one projector $P_i$ (or equivalently, each vector $|v_i\>$ in the contextuality set $\mathcal{V}$) is translated to a binary observable defined as $A_{i} \equiv \mathbb{I} - 2P_i$, with eigenvalues $\{+1, -1\}$. The compatibility relations among these observables arise directly from the orthogonality relations of the corresponding projectors. Every set of compatible observables defines a context. Compatible observables can in principle be jointly measured, and their measurement outcomes remain invariant under different orders of sequential implementation.

\section{Proof of Theorem 1}\label{app_proof_QuantumValue}
For the convenience of the reader, we restate Theorem 1 from the main text as follows:
\begin{theorem}
The weighted independence number of the graph $(\mathcal{G}_d,\bm{w})$ is $\alpha(\mathcal{G}_d,\bm{w}) = 2d + 1$ for any $d \geq 3$, where vertices $v_{i,1} \ \forall i$ are assigned weight $2$ and all other vertices are assigned weight $1$. The weighted Lovász theta number for these weights is $\theta(\mathcal{G}_d,\bm{w}) = 2d+2$ for $d \geq 9$, with the value for $3 \leq d \leq 8$ being given as in the Table~\ref{table1}. The weighted fractional packing number is $\alpha^*(\mathcal{G}_d,\bm{w}) = 2d+2$ for $d \geq 3$. Furthermore, for any quantum correlation achieving the weighted Lovász theta number, there exists a measurement setting (a context) $c^*$ that produces $d$ outcomes, each occurring with a uniform probability of $1/d$.
\end{theorem}

	\begin{proof}
	We have already established the relationships between the NCHV and quantum bounds of the non-contextual expression $\mathcal{I}_d$ with the weighted independence number and weighted Lovász theta number of the corresponding graph $(\mathcal{G}_d, \bm{w})$, we now focus on calculating these parameters.
	
	The independence number of the weighted orthogonality graph $(\mathcal{G}_d, \bm{w})$ is given by $\alpha(\mathcal{G}_d, \bm{w}) = 2d + 1$. This follows from the fact that $\mathcal{G}_d$ contains $d$ disjoint $C_5$ cycles (hereafter denoted by $C_{5,i}$), each of which has an independence number of 2 when weights are not considered. Since the vertices $v_{i,1}$ have a weight of 2, higher than the other vertices with weight 1, and form a clique in $\mathcal{G}_d$, the maximal independent set must include exactly one vertex from $\{v_{i,1}\}_{i=1}^d$. Therefore, the weighted independence number is generally $2d - 1 + 2 = 2d + 1$.

    The Lovász theta number for the graph $\mathcal{G}_d$ with vertex weights $\bm{w}$ is defined as $\theta(\mathcal{G}_d, \bm{w}) := \max_{\{|v_{i,j}\>, |\psi\>\}} \sum_{i=1}^d \sum_{j=1}^5 w(v_{i,j}) |\<v_{i,j}|\psi\>|^2$ where $\{|v_{i,j}\>\}$ is an orthogonal representation (OR) of the graph $\mathcal{G}_d$. We claim that there exists a set of vectors $\{|{v}_{i,j}\>\}$ and a state $|{\psi}\>$ that achieve the Lovász theta number with the property $|\<{v}_{i,j}|{\psi}\>|^2 = |\<{v}_{i',j}|{\psi}\>|^2, \forall i,i',j$. This is justified by noting that the Lovász theta number for a graph $G = (V, E)$ with weights $\bm{w}$ can be computed using the following semidefinite program (SDP)~\cite{grotschel2012geometric}:
	\begin{equation}\label{lovaszSDP}
		\begin{split}
			\max_X & \quad \sum_{u=1}^d w(u) X_{0,u} \\
			\text{s.t.} & \quad X_{u,v} = 0, \forall (u, v) \in E, \\
			& \quad X_{0,0} = 1, \\
			& \quad X_{u,u} = X_{0,u}, \forall 1 \leq u \leq |V|, \\
			& \quad X \in \mathbb{S}_{+}^{|V| + 1}.
		\end{split}
	\end{equation}
    From an optimal solution $X^*$ of the above SDP, the contribution of each vertex to the Lovász theta number can be directly determined as $|\<{u} |{\psi}\>|^2 = X^*_{0,u}$.

    For the graph $\mathcal{G}_d$ under consideration, suppose that an optimal Lovász theta SDP solution $\widetilde{X}$ yields $\widetilde{X}_{0,(i,j)} \neq \widetilde{X}_{0,(i',j)}$ for a specific $j$ and for indices $i$ and $i'$. There exists a graph isomorphism $f$ from $\mathcal{G}_d$ to itself that interchanges the vertices $v_{i,j}$ and $v_{i',j}$ for all $j$. Denote the resulting matrix after interchanging the corresponding pairs of rows and columns of $\widetilde{X}$ as $\widetilde{X}^{(i,i')}$. Due to the symmetry of $\mathcal{G}_d$, $\widetilde{X}^{(i,i')}$ is also an optimal solution to the original SDP, satisfying $\widetilde{X}^{(i,i')}_{0,(i,j)} = \widetilde{X}_{0,(i',j)}$ and $\widetilde{X}^{(i,i')}_{0,(i',j)} = \widetilde{X}_{0,(i,j)}$ for all $j$. 

    We can update the optimal solution as $\widetilde{X} = \frac{1}{2}(\widetilde{X} + \widetilde{X}^{(i,i')})$, this ensures that $\widetilde{X}_{0,(i,j)} = \widetilde{X}_{0,(i',j)}$ for all $j$. If there exists an $i''$ such that $\widetilde{X}_{0,(i,j)} \neq \widetilde{X}_{0,(i'',j)}$ for some $j$, we update the optimal solution to $\widetilde{X} = \frac{1}{3}(\widetilde{X} + \widetilde{X}^{(i,i'')} + \widetilde{X}^{(i',i'')})$, which can ensure that $\widetilde{X}_{0,(i,j)} = \widetilde{X}_{0,(i',j)} = \widetilde{X}_{0,(i'',j)}$ for all $j$. This iterative process can address other offending indices of $i$, maintaining the established relationships $\widetilde{X}_{0,(i,j)} = \widetilde{X}_{0,(i',j)}$ for previously adjusted indices, as the isomorphism and updates ensure subsequent adjustments preserve these equalities.

	Since the vertices $v_{i,1}$ form a clique of size $d$, the sum of their contributions to the usual Lovász theta number is at most $1$, i.e., $\sum_{i=1}^d |\<{v}_{i,1}|{\psi}\>|^2 \leq 1$. Each of these contributions is equal and therefore at most $1/d$. Without loss of generality, we assume $|\<{v}_{i,1}|{\psi}\>|^2 = 1/t$ for all $i$, where $t \geq d$. In the following, we will demonstrate that $t = d$ when the Lovász theta number is achieved.
	
	Based on the above discussion, it is clear that the Lovász theta number for the graph $\mathcal{G}_d$ is equivalent to $d$ times the Lovász theta number of a single $C_{5,i}$ under the condition $|\<{v}_{i,1}|{\psi}\>|^2 = 1/t$. This equivalence holds because, apart from $v_{i,1}$, the other vertices in $C_{5,i}$ do not share neighbors with any vertices outside of $C_{5,i}$. Thus, from the solution of the Lovász theta SDP for the graph $\mathcal{G}_d$, we can directly obtain the Lovász theta number for $C_{5,i}$ under the condition $|\<{v}_{i,1}|{\psi}\>|^2 = 1/t$.
    Focusing on the Lovász theta number of $C_{5,1}$ under the condition $|\<{v}_{1,1}|{\psi}\>|^2 = 1/t$, it is known that the minimum dimension of an orthogonal representation (OR) for $C_5$ is $3$. Without loss of generality, let $|v_{1,1}\>:=(1,0,0)^T$, we can then parameterize its two orthogonal vectors $|v_{1,2}\>,$ and $ |v_{1,5}\>$ using parameters $\gamma_1,\gamma_2\in[-\pi/2,\pi/2)$ as follows:
    \begin{equation}
        |v_{1,2}\> = \begin{pmatrix}0 \\ \cos \gamma_1 \\ \sin \gamma_1\end{pmatrix}, \quad  |v_{1,5}\> = \begin{pmatrix}0 \\ \cos \gamma_2 \\ \sin \gamma_2\end{pmatrix}.
    \end{equation}
    Since $|v_{1,3}\>$ is orthogonal to $|v_{1,2}\>$ and $|v_{1,4}\>$ is orthogonal to $|v_{1,5}\>$, we parameterize $|v_{1,3}\>$ and $|v_{1,4}\>$ using parameters $\phi_1,\phi_2\in[-\pi/2,\pi/2)$ as follows:
    \begin{equation}
        |v_{1,3}\> = \begin{pmatrix}\sin \phi_1 \\ -\cos \phi_1 \sin \gamma_1 \\ \cos \phi_1 \cos \gamma_1\end{pmatrix},\quad |v_{1,4}\> = \begin{pmatrix}\sin \phi_2 \\ -\cos \phi_2 \sin \gamma_2 \\ \cos \phi_2 \cos \gamma_2\end{pmatrix}.
    \end{equation}
    Lastly, we parameterize $|{\psi}\>$ by $\gamma\in[-\pi/2,\pi/2)$ and $t$ as follows such that the condition $|\<{v}_{1,1}|{\psi}\>|^2 = 1/t$ holds true:
    \begin{equation}
        |{\psi}\> = \left(\frac{1}{\sqrt{t}}, \sqrt{\frac{t-1}{t}} \cos \gamma, \sqrt{\frac{t-1}{t}} \sin \gamma \right)^T.
    \end{equation}
    Because of the orthogonality relation between $|v_{1,3}\>$ and $|v_{1,4}\>$, we have that:
	\begin{equation}\label{ORcon}
		\<v_{1,3}|v_{1,4}\>=\sin \phi_1 \sin \phi_2 + \cos \phi_1 \cos \phi_2 \cos (\gamma_1 - \gamma_2) = 0.
	\end{equation}

    By symmetry, there exists a solution to the SDP such that the contribution of $v_{1,2}$  to the Lovász theta number equals that of $v_{1,5}$, and similarly for $v_{1,3}$ and $v_{1,4}$. This implies that there exist values for the parameters $\gamma,\gamma_1,\gamma_2,\phi_1,\phi_2$ such that the two equations $|\<{v}_{1,2}|{\psi}\>|^2 = |\<{v}_{1,5}|{\psi}\>|^2$ and $|\<{v}_{1,3}|{\psi}\>|^2 = |\<{v}_{1,4}|{\psi}\>|^2$ hold. Nontrivial solutions to these equations result in $\gamma = \frac{\gamma_1 + \gamma_2}{2}$ and $\phi_1 = -\phi_2 =: \phi$. Under these conditions, defining $k := \frac{\gamma_2 - \gamma_1}{2}$, Eq.~\eqref{ORcon} simplifies to
    \begin{equation}
    \cos^2 \phi = \frac{1}{2 \cos^2 k}, \quad \sin^2 \phi = \frac{2 \cos^2 k - 1}{2 \cos^2 k}.
    \end{equation}
    The corresponding weighted Lovász theta number is then calculated as:
	\begin{equation}\label{theta_c51}
		\begin{split}
		  \theta[(C_{5,1},\bm{w}),t] :&=\max_{\text{OR of }C_{5,1}\{|v_{1,j}\>\}, |\psi\>} \sum_{j=1}^5 w_{1,j}|\<v_{1,j}|\psi\>|^2\\
			&=\max_{\gamma,\gamma_1,\gamma_2,\phi_1,\phi_2} 2\cdot \frac{1}{t}+\frac{t-1}{t} \sum_{i=1}^{2}\cos^2(\gamma-\gamma_i)+\sum_{i=1}^{2} \big(\frac{1}{\sqrt{t}}\sin\phi_i+\sqrt{\frac{t-1}{t}}\cos\phi_i\sin(\gamma-\gamma_i) \big)^2\\
			&=\max_{k,\phi} 2\cdot \frac{1}{t}+2\cdot \frac{t-1}{t} \cos^2 k+2\cdot \big(\frac{1}{\sqrt{t}}\sin\phi+\sqrt{\frac{t-1}{t}}\cos\phi\sin k \big)^2\\
			&=\max_{k} 2\cdot \frac{1}{t}+2\cdot \frac{t-1}{t} \cos^2 k+\frac{1}{t\cos^2 k}\cdot
			\big(\sqrt{2\cos^2 k-1}+\sqrt{t-1}\sin k \big)^2\\
		\end{split}
	\end{equation}
    Denoting $x:=\cos^2 k$, we have $x\in [1/2,1]$ and the expression of Eq.~\eqref{theta_c51} reduces to 
    \begin{equation}
        \theta[(C_{5,1},\bm{w}),t]=\max_{x} \frac{5}{t}-1 + \frac{2(t-1)}{t}\cdot x +\frac{t-2}{t}\cdot \frac{1}{x}+\frac{2\sqrt{t-1}}{t}\cdot \frac{\sqrt{-2x^2+3x-1}}{x}.
    \end{equation}
	
	Assume that for a given $t$, the solution to the maximization problem in the last line above is $x^*_t$ with $x^*_t \in [1/2, 1]$. According to the Envelope theorem~\cite{milgrom2002envelope}, suppose $f(a) := \max_b g(a, b)$, where $g(a, b)$ is a real-valued, continuously differentiable function in both $a$ and $b$, and there exists an optimal point $b^*(a)$ such that $f(a) = g(a, b^*(a))$. Then, the derivative of $f(a)$ with respect to $a$ is given by $\frac{d }{d a} f = \frac{\partial g}{\partial a} \big|_{b=b^*(a)}$. 
	Thus, we have 
	\begin{equation}
		\frac{d }{d t} \theta[(C_{5,1}, \bm{w}), t] = -\frac{5 - 2x^*_t - \frac{2}{x^*_t}}{t^2} + \frac{(t - 1)^{-1/2}t - 2(t - 1)^{1/2}}{t^2} \cdot \frac{\sqrt{-2{x^*_t}^2 + 3x^*_t - 1}}{x^*_t} \leq 0.
	\end{equation}
	The above inequality holds because $t \geq d \geq 3$ implies $(t - 1)^{-1/2}t - 2(t - 1)^{1/2} < 0$, and $x^*_t \in [1/2, 1]$ implies $5 - 2x^*_t - \frac{2}{x^*_t} \geq 0$. Therefore, $\theta[(C_{5,1}, \bm{w}), t]$ is a monotonically decreasing function of $t$, which means that, in order to achieve the Lovász theta number of the graph $\mathcal{G}_d$, we must have $t = d$ and hence $|\<{v}_{i,1}|{\psi}\>|^2 = 1/d$ for all $i$.
	
	In general, we have that the Lovász theta number of $\mathcal{G}_d$ is given by:
	\begin{equation}
		\theta(\mathcal{G}_d, \bm{w}) = d \cdot \theta[(C_{5,1}, \bm{w}), d] = d \cdot \left( \max_{x} \frac{5}{d} - 1 + \frac{2(d - 1)}{d} \cdot x + \frac{d - 2}{d} \cdot \frac{1}{x} + \frac{2\sqrt{d - 1}}{d} \cdot \frac{\sqrt{-2x^2 + 3x - 1}}{x} \right).
	\end{equation}
	Taking the derivative of the above equation with respect to $x$ and setting the result to zero, we find that for $d \geq 9$, the solutions of the differential equation are $x^*_t = \frac{3}{4} \pm \frac{1}{4} \sqrt{\frac{d - 9}{d - 1}}$, which gives $\theta[(C_{5,1}, \bm{w}), d] = 2 + \frac{2}{d}$ and $\theta(\mathcal{G}_d, \bm{w}) = 2d + 2$.
	For $3 \leq d \leq 9$, the solution of the differential equation is more complex, so we provide the numerical solutions in the following table.
	
	\begin{table}[H]
		\centering
		\begin{tabular}{|c||c|c|c|c|c|c|c|}
			\hline
			dimension $d$ & 3 & 4 & 5 & 6 & 7 & 8 & $d\geq 9$ \\
			\hline
			\hline
			$\theta[(C_{5,1},\bm{w}),d]$ & 2.5585 & 2.4508 & 2.3774 & 2.3237 & 2.2823 & 2.2493 & $2+\frac{2}{d}$\\
			\hline
			$\theta(\mathcal{G}_d,\bm{w})$ & 7.6753 & 9.8030 & 11.8869 & 13.9419 &  15.9762& 17.9944 & $2d+2$ \\
			\hline
			$\alpha(\mathcal{G}_d,\bm{w})$ & 7 & 9 & 11 & 13 & 15 & 17 & $2d+1$ \\
			\hline
			\hline
				$\alpha^*(\mathcal{G}_d,\bm{w})$ & 8 & 10 & 12 & 14 & 16 & 18 & $2d+2$ \\
				\hline
		\end{tabular}
		\caption{The weighted Lovász theta number of $C_5$ under the condition that one vertex contributes $1/d$ to the Lovász theta function, alongside the weighted independence number, weighted Lovász theta number and weighted fractional packing number of the orthogonality graphs $\mathcal{G}_d$.}
		\label{tab:my_label}
	\end{table}
	Finally, we compute the weighted fractional packing number of the graph $\alpha^*(\mathcal{G}_d,\bm{w})$. This is given by the following linear program:
	\begin{eqnarray}
	\alpha^*(\mathcal{G}_d,\bm{w}) &=& \max \sum_{v \in V(\mathcal{G}_d)} w_v p_v, \nonumber \\
	\text{s.t.} && \; \sum_{v \in C} p_v \leq 1, \quad \forall C \in \Gamma(\mathcal{G}_d), \nonumber \\
	&& \; 0 \leq p_v \leq 1, \quad \forall v \in V(\mathcal{G}_d), \nonumber \\	
	\end{eqnarray}
where $\Gamma(\mathcal{G}_d)$ denotes the set of all cliques in the graph $\mathcal{G}_d$. One can now find a clique cover in the graph given by the central clique $\{v_{1,1}, \ldots, v_{d,1}\}$ along with the edges $\{v_{i,2}, v_{i,3}\}$ and $\{v_{i,4}, v_{i,5}\}$ for $i =1,\ldots, d$. With the weight $2$ for the central clique and the constraint $\sum_{v \in C} p_v \leq 1$ for each clique, we see that the upper bound for the weighted fractional packing number is $2d+2$ which is tight since it can be achieved by an assignment $p_v = 1/2$ for the vertices $\{v_{i,2}, v_{i,3}\}$ and $\{v_{i,4}, v_{i,5}\}$ for $i =1,\ldots, d$ and $p_v = 1/d$ for the vertices in $\{v_{1,1}, \ldots, v_{d,1}\}$. 
\end{proof}

\section{Quantum Violations of $\epsilon$-ONC Models for arbitrary relaxation --Proof of Theorem 2}\label{proof_theo_relax}

In the main text, we motivated the NCHV model with relaxing strict orthogonality relations into $\epsilon$-orthogonality and discussed its graph-theoretic representations. In this section, we elaborate on how to calculate the classical bound of a non-contextuality expression in the relaxed model using graph parameters, and we then prove Theorem 2 in the main text.

A general non-contextuality expression represented by the weighted orthogonality graph $G = (V, E)$ with weights $\bm{w}$ is given by 
\begin{eqnarray}\label{app_nc}
    \sum_{v \in V} w(v) p_v = \sum_{C \in \Gamma} \sum_{v \in C} \frac{w(v)}{n_v} \cdot \text{Pr}[v|C],
\end{eqnarray}
where $\Gamma$ denotes the set of maximal cliques of $G$, $n_v$ is the number of maximal cliques in $\Gamma$ containing $v \in V$ and $w(v)$ is the weight of $v$ in $\bm{w}$. The term $\text{Pr}[v|C]$ represents the probability of obtaining the outcome corresponding to $v$ when measuring in the context $C$. In quantum theory, the probability of obtaining outcome $v$ is independent of the context, i.e., $\text{Pr}[v|C] = |\langle \psi| v \rangle|^2$, regardless of the context $C$.

In the main text, we explained how to construct the $\epsilon$-orthogonality graph $G_{\epsilon} = (V_{\epsilon}, E_{\epsilon}, \widetilde{E}_{\epsilon})$ with vertex weights $\bm{w}_{\epsilon}$ from a given orthogonality graph $G = (V, E)$ with vertex weights $\bm{w}$. We define $V_{\epsilon} := \{v^{(C)} \mid v \in V, C \in \Gamma\}$. For each edge $(v, u) \in E$ and contexts $C, C' \in \Gamma$ where $v \in C$ and $u \in C'$, $G_{\epsilon}$ has an edge between $v^{(C)}$ and $u^{(C')}$: $(v^{(C)}, u^{(C')}) \in E_{\epsilon}$ when $C = C'$, and $(v^{(C)}, u^{(C')}) \in \widetilde{E}_{\epsilon}$ when $C \neq C'$. The vertex weights $\bm{w}_{\epsilon}$ are defined as $w_{\epsilon}(v^{(C)}) := w(v)/n_v$. We also define two related simple graphs associated with $G_{\epsilon}$: $G' := (V_{\epsilon}, E_{\epsilon})$, which consists of the edges of $G_{\epsilon}$ (forming disjoint maximal cliques), and $G'' := (V_{\epsilon}, E_{\epsilon} \cup \widetilde{E}_{\epsilon})$, which treats the $\epsilon$-edges as ``standard'' edges. We claim that the maximum value of the non-contextuality inequality in the $\epsilon$-ONC model, $\alpha_{\epsilon}(G_{\epsilon}, \bm{w}_{\epsilon})$, is bounded by:
\begin{equation}\label{ep_alpha_1}
\alpha_{\epsilon}(G_{\epsilon}, \bm{w}_{\epsilon}) \leq \epsilon \cdot \alpha(G', \bm{w}_{\epsilon}) + (1 - \epsilon) \cdot \alpha(G'', \bm{w}_{\epsilon}).
\end{equation}
We now proceed to prove this statement by maximizing the non-contextuality expression given in Eq.~\eqref{app_nc} over the $\epsilon$-ONC models:
\begin{eqnarray}
    \sum_{C \in \Gamma} \sum_{v \in C} \; \frac{w(v)}{n_v} \cdot \text{Pr}[v|C] &=& \sum_{C \in \Gamma} \sum_{v \in C} \; \frac{w(v)}{n_v} \cdot f_v^{C}
 \end{eqnarray}
 with $f_v^C\in \{0,1\}$. Consider a set of non-contextual hidden variables $\lambda$ such that $\left(f_v^{C}\right)_{\lambda} \in \{0,1\}$, with $q_{\lambda} \geq 0$ and $\sum_{\lambda} q_{\lambda} = 1$. Let $\Omega$ denote the set of hidden variable models that assign different values to the ``same'' vertex $v$ appearing in different contexts, i.e., 
\begin{equation}
    \Omega := \left\{ \lambda : \exists v \in C \cap C' \; \text{s.t.} \; \left(f_v^{C}\right)_{\lambda} \neq \left(f_v^{C'}\right)_{\lambda} \; \text{for} \; C, C' \in \Gamma \right\}. 
\end{equation}
For $\lambda \notin \Omega$, it is clear that the maximum value of a non-contextuality expression is simply $\alpha(G, \bm{w})$, the weighted independence number of the original graph $G$ with weights $\bm{w}$, since such $\lambda$ simply correspond to the NCHV models for $G$. This is also equal to $\alpha(G'', \bm{w}_{\epsilon})$, as shown in Lemma \ref{lemma_independent_relax}. For $\lambda \in \Omega$, each of the maximal clique conditions $\sum_{v \in C} p_v \leq 1$ can be saturated, leading to a value of the non-contextuality expression equal to $\alpha(G', \bm{w}_{\epsilon})$. Recall that, based on the definition of $G_{\epsilon}$, the graph $G'$ consists solely of disjoint maximal cliques.
We thus write the $\epsilon$-ONC value as
 \begin{eqnarray}
    \sum_{C \in \Gamma}\sum_{v \in C} \; \frac{w(v)}{n_v} \cdot \text{Pr}[v|C] &=& \sum_{\lambda} q_{\lambda} \sum_{C \in \Gamma}\sum_{v \in C} \; \frac{w(v)}{n_v} \cdot \left(f_v^{C} \right)_{\lambda} \nonumber \\
    &=& \sum_{\lambda \in \Omega} q_{\lambda} \sum_{C \in \Gamma}\sum_{v \in C} \; \frac{w(v)}{n_v} \cdot \left(f_v^{C} \right)_{\lambda} + \sum_{\lambda \notin \Omega} q_{\lambda} \sum_{C \in \Gamma}\sum_{v \in C} \; \frac{w(v)}{n_v} \cdot \left(f_v^{C} \right)_{\lambda} \nonumber \\
    &\leq& \sum_{\lambda \in \Omega} q_{\lambda} \alpha(G', \bm{w}_{\epsilon}) + \sum_{\lambda \notin \Omega} q_{\lambda} \alpha(G, \bm{w}) \nonumber \\
    &\leq& \epsilon \cdot \alpha(G', \bm{w}_{\epsilon}) + (1 - \epsilon) \cdot \alpha(G, \bm{w}),
 \end{eqnarray} 
where in the last inequality we have used the fact that $\alpha(G, \bm{w}) \leq \alpha(G', \bm{w}_{\epsilon})$ which comes from the fact every maximal clique can have at most one vertex in the maximum independent set.      

 \begin{lemma}
    Let $(G, \bm{w})$ be a weighted graph and $(G_{\epsilon}, \bm{w}_{\epsilon})$ its corresponding weighted $\epsilon$-graph as described in the main text. Denote by $G'' = (V_{\epsilon}, E_{\epsilon} \cup \widetilde{E}_{\epsilon})$ the simple graph obtained from $(G_{\epsilon}, \bm{w}_{\epsilon})$. Then the weighted independence number satisfies $\alpha(G, \bm{w}) = \alpha(G'', \bm{w}_{\epsilon})$.
    \label{lemma_independent_relax}
\end{lemma}

\begin{proof}
    We prove the equality by showing both $\alpha(G, \bm{w}) \leq \alpha(G'', \bm{w}_{\epsilon})$ and $\alpha(G, \bm{w}) \geq \alpha(G'', \bm{w}_{\epsilon})$.

    Let $I \subseteq V$ be a maximum weighted independent set of $G$ with weight assignment $\bm{w}$, achieving $\alpha(G, \bm{w})= \sum_{v \in I} w(v)$. For each $v\in V$, denote $\Gamma_{v}\subset \Gamma$ as the set of maximal cliques containing $v$. By construction in the $\epsilon$-graph $G_{\epsilon}$ (and thus in $G''$), the set $I' := \{v^{(C)}|v\in I,C \in \Gamma_v\}$ is an independent set of $G''$. This is because $v^{(C)}$ and $v^{(C')}$, for all $v\in V$ and $C,C' \in \Gamma_v$, are not adjacent in $G''$. Moreover, $v^{(C)}$ and $v'^{(C')}$, for all $v,v' \in I$, $C \in \Gamma_v$, and $C' \in \Gamma_{v'}$, are not adjacent in $G''$; otherwise, $v$ and $v'$ would be adjacent in $G$. Furthermore, 
    \begin{equation}
        \sum_{v^{(C)} \in I'} w_{\epsilon}(v^{(C)}) = \sum_{v \in I} \sum_{C \in \Gamma_v} \frac{w(v)}{n_{v}} = \sum_{v \in I} w(v) = \alpha(G, \bm{w}).
    \end{equation} 
    Thus, the weighted independence number of $G''$ under the weight assignment $\bm{w}_{\epsilon}$ is at least $\alpha(G, \bm{w}).$

    Conversely, let $J \subseteq V_{\epsilon}$ be a maximum weighted independent set in $G''$ (under the weight assignment $\bm{w}_{\epsilon}$), so $\alpha(G'', \bm{w}_{\epsilon}) = \sum_{v^{(C)}\in J} w_{\epsilon}(v^{(C)})$. For the same reason, the set $J' := \{v | v^{(C)} \in J\}$ is an independent set of $G$. Similarly, we have 
    \begin{equation}
        \sum_{v \in J'} w(v) = \sum_{v \in J'} \sum_{C \in \Gamma_v} \frac{w(v)}{n_{v}} = \sum_{v^{(C)}\in\{v^{(C)}|v \in J', C \in \Gamma_{v}\}} w_{\epsilon}(v^{(C)}) \geq \alpha(G'', \bm{w}_{\epsilon}).
    \end{equation}
    Thus, the weighted independence number of $G$ under the weight assignment $\bm{w}$ is at least $\alpha(G'', \bm{w}_{\epsilon}).$

    Combining these two statements yields $\alpha(G, \bm{w}) = \alpha(G'', \bm{w}_{\epsilon})$, as claimed.
\end{proof}

We now proceed to prove the main theorem of this section, which states that quantum mechanics cannot be described by the $\epsilon$-ONC model for any arbitrary $\epsilon \in [0, 1)$. For the convenience of the reader, we restate Theorem 2 from the main text as follows:
        \begin{theorem}
        For an arbitrary $\epsilon \in [0, 1)$, there exists a non-contextuality inequality and a corresponding optimal quantum correlation that certifies contextuality against $\epsilon$-ONC models.
	\end{theorem}
	\begin{proof}
    To prove this theorem, we demonstrate that for any $\epsilon \in [0, 1)$, there always exists a non-contextual expression (i.e., the weighted orthogonality graph $(G,\bm{w})$, using graph representations) that reveals the gap between the quantum correlation set and the $\epsilon$-ONC correlation set. Specifically, we show that the quantum bound of the non-contextual expression surpasses $\alpha_{\epsilon}\left(G_{\epsilon}, \bm{w}_{\epsilon}\right)$, the $\epsilon$-ONC bound for the expression.

    The particular graph of interest is the odd cycle $C_n$, with each vertex weighted by $1$, where $n \geq 3$ is an odd number. The odd cycle $C_n$ consists of $n$ vertices labeled as $V = \{0, 1, \ldots, n-1\}$ and has $n$ edges $(i, i+1)$ for $i \in \{0, 1, \ldots, n-1\}$, with addition taken modulo $n$. Each maximal clique contains 2 vertices, and each vertex belongs to two cliques. An example of the odd cycle $C_5$ is shown in Fig.~\ref{c5_relax}. The odd cycle graph $C_n$ has been extensively studied in the literature, and it is known that its independence number is $\alpha(C_n) = \frac{n-1}{2}$. Additionally, the quantum bound of the corresponding non-contextual expression is given by the Lovász theta number for $C_n$~\cite{lovasz1979shannon}, which is:
    \begin{equation}
    \theta(C_n) = \frac{n \cos{\left(\frac{\pi}{n}\right)}}{1 + \cos{\left(\frac{\pi}{n}\right)}}.
    \end{equation}
    This value is achieved using an orthonormal representation in 3-dimensional space.

    We construct the corresponding $\epsilon$-graph, denoted as $C_{n,\epsilon}$, for an odd cycle $C_n$. This graph contains $2n$ vertices, each with weight $\frac{1}{2}$, $n$ edges, and $2n$ $\epsilon$-edges. Fig.~\ref{c5_relax} illustrates the $\epsilon$-graph for $C_5$, showing how the original odd cycle is transformed by incorporating $\epsilon$-edges.

\begin{figure}[H]
    \centering
    \subfigure[]{
        \begin{minipage}[t]{0.17\textwidth}
            \centering
            \includegraphics[width=1\textwidth]{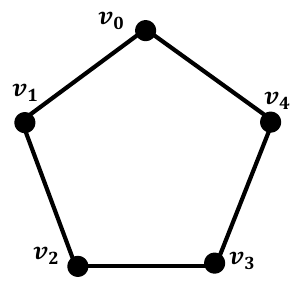}
        \end{minipage}
    }
    \qquad \qquad
    \subfigure[]{
        \begin{minipage}[t]{0.17\textwidth}
            \centering
            \includegraphics[width=1\textwidth]{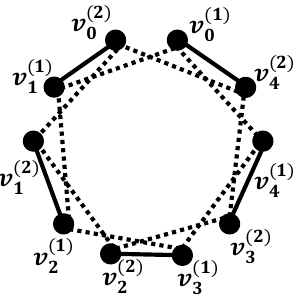}
        \end{minipage}
    }
    \caption{An example of (a) the orthogonality graph for  $C_5$ and its $\epsilon$-orthogonality graph (b) $C_{5,\epsilon}$, where solid edges denote orthogonality and dotted lines denote $\epsilon$-orthogonality. Note that the weight of each vertex in (a) is $1$, while the weight of each vertex in (b) is $\frac{1}{2}$.} 
    \label{c5_relax}
\end{figure}

The $\epsilon$-independence number $\alpha_{\epsilon}(C_{n,\epsilon}, \bm{w}_{\epsilon})$ is given by the convex combination of the independence numbers of two related graphs, $C_{n}'$ and $C_{n}''$, with coefficients $\epsilon$ and $1-\epsilon$, respectively. The graph $C_{n}'$ is a copy of $C_{n,\epsilon}$ with all $\epsilon$-edges removed, resulting in $n$ disjoint edges. Thus, we find that $\alpha(C_{n}', \bm{w}_{\epsilon}) = \frac{n}{2}$. The graph $C_{n}''$ is a copy of $C_{n,\epsilon}$ where all $\epsilon$-edges are treated as ``standard'' edges. By Lemma~\ref{lemma_independent_relax}, we have $\alpha(C_{n}'', \bm{w}_{\epsilon}) = \alpha(C_n) = \frac{n-1}{2}$. Therefore, the $\epsilon$-independence number of $C_{n,\epsilon}$ is given by:
\begin{equation}
    \alpha_{\epsilon}(C_{n,\epsilon}, \bm{w}_{\epsilon}) = \epsilon \cdot \frac{n}{2} + (1-\epsilon) \cdot \frac{n-1}{2} = \frac{n-1+\epsilon}{2}.
\end{equation}

Thus, for any odd cycle $C_n$, the threshold of $\epsilon$ can be derived from the inequality $\alpha_{\epsilon}(C_{n,\epsilon}, \bm{w}_{\epsilon}) \leq \theta(C_n)$. Therefore, we have:
\begin{equation}
    \epsilon \leq 1-n+\frac{2n\cos{\left(\frac{\pi}{n}\right)}}{1+\cos{\left(\frac{\pi}{n}\right)}} = 1+\frac{2n\left(\cos^2{\frac{\pi}{2n}} - \sin^2{\frac{\pi}{2n}}\right) - 2n\cos^2{\frac{\pi}{2n}}}{2\cos^2{\frac{\pi}{2n}}}
    = 1-n\tan^2{\frac{\pi}{2n}} \leq 1-n\cdot \left(\frac{\pi}{2n}\right)^2 = 1 - \frac{\pi^2}{4n}.
\end{equation}
In other words, for any $\epsilon \in [0, 1)$, we find that $\alpha_{\epsilon}(C_{n,\epsilon}, \bm{w}_{\epsilon}) \leq \theta(C_n$) when the length of the odd cycle satisfies $n \geq \lceil\frac{\pi^2}{4(1-\epsilon)}\rceil$.

\end{proof}
	
		\section{Proof of Lemma 1}\label{app_proof_lemma}
For the convenience of the reader, we restate Lemma 1 from the main text as follows:
\begin{lemma}\label{lemma_theta}
		For any $\epsilon$-graph $G_{\epsilon}$ with vertex weights $\bm{w}_{\epsilon}$, the $\epsilon$-Lovász theta number is lower-bounded by the SDP relaxation of the $\epsilon$-independence number, i.e., $\theta_{\epsilon}(G_{\epsilon}, \bm{w}_{\epsilon}) \geq \theta'_{\epsilon}(G_{\epsilon}, \bm{w}_{\epsilon})$, where
		\begin{equation}
					\label{app_theta_relax}
			\begin{split}
				\theta'_{\epsilon}(G_{\epsilon}, \bm{w}_{\epsilon}) := \max_{X, Y} & \; \sqrt{\epsilon}(\sqrt{\bm{w}_{\epsilon}}^T X \sqrt{\bm{w}_{\epsilon}}) + (1 - \sqrt{\epsilon})  (\sqrt{\bm{w}_{\epsilon}}^T Y \sqrt{\bm{w}_{\epsilon}}) \\
				\text{s.t.} & \quad X_{vu} = 0, \; \; \forall (v, u) \in E_{\epsilon}, \\
				& \quad Y_{vu} = 0, \; \; \forall (v, u) \in E_{\epsilon} \cup \widetilde{E}_{\epsilon}, \\
				& \quad X_{vv} = Y_{vv}, \; \; \forall v \in V_{\epsilon}, \\
				& \quad \Tr(X) = 1, \; \; \Tr(Y) = 1, \\
				& \quad X \succeq 0, \;\; Y \succeq 0. 		
			\end{split}
		\end{equation}
	\end{lemma}

\begin{proof}
Suppose that $X^{*}$ and $Y^{*}$ are the maximizers of the SDP problem in Eq.~\eqref{app_theta_relax}. We define the matrix $Z$ as:
\begin{equation}
	Z := \sqrt{\epsilon}\cdot X^{*} + (1-\sqrt{\epsilon})\cdot Y^{*}.
\end{equation}
This construction guarantees that $Z \succeq 0$, $\Tr Z = 1$, $Z_{vu} = 0$ for all $(v,u) \in E$, and $Z_{vu} = \sqrt{\epsilon}\cdot X^*_{vu}$ for all $(v,u) \in \widetilde{E}_{\epsilon}$. Additionally, we have $Z_{vv} = X^*_{vv} = Y^*_{vv}$ for all $v \in V$. Since $Z$ is positive semi-definite, its Cholesky decomposition yields $|V|$ (unnormalized) vectors $\{|v^{*}\>\}, v \in V$ such that $Z_{vu} = \<v^*|u^*\>$.

Next, we define the following unit vectors:
\begin{equation}\label{app_def_v}
	\begin{split}
		|{v}\> & := \frac{|v^{*}\>}{\||v^{*}\>\|} \quad \forall v \in V,\\
		|{\psi}\> & := \frac{\sum_{v \in V} \sqrt{w_{\epsilon}(v)} |v^{*}\>}{\|\sum_{v \in V} \sqrt{w_{\epsilon}(v)} |v^{*}\>\|}.
	\end{split}
\end{equation}

The set of vectors $\{|{v}\>\}, v \in V$ is an orthonormal representation of $G'$, satisfying the strict orthogonality constraints for edges in $E_{\epsilon}$. Additionally, they meet the $\epsilon$-orthogonality conditions for the $\epsilon$-edges in $\widetilde{E}_{\epsilon}$. For each $(v,u) \in \widetilde{E}_{\epsilon}$, we have:
\begin{equation}
	\begin{split}
		\left|\<v| u\>\right| = \frac{\left|\<v^*| u^*\>\right|}{\||v^*\>\| \cdot \||u^*\>\|}= \frac{\left|Z_{vu}\right|}{\sqrt{Z_{vv}} \cdot \sqrt{Z_{uu}}} = \frac{\left|\sqrt{\epsilon} \cdot X^*_{vu}\right|}{\sqrt{X^*_{vv}} \cdot \sqrt{X^*_{uu}}} \leq \sqrt{\epsilon}.
	\end{split}
\end{equation}

The last inequality holds because $X^*$ is positive semi-definite and can also be Cholesky decomposed into a set of vectors $\{|x_v^*\>\}, v \in V$, where $\<x_v^*| x_u^*\> = X^*_{vu}$. Thus, we have:
\begin{equation}
	\frac{\left|X^*_{vu}\right|}{\sqrt{X^*_{vv}} \cdot \sqrt{X^*_{uu}}} = \frac{\left|\<x_v^*|x_u^*\>\right|}{\||x_v^*\>\| \cdot \||x_u^*\>\|} = \left|\<{x_v}|{x_u}\>\right| \leq 1,
\end{equation}
where $\{|{x_v}\>\}, v \in V$ are normalized vectors corresponding to $\{|x_v^*\>\}$.

For the set of vectors $\{|v\>\}, v \in V$ and $|\psi\>$ defined in Eq.~\eqref{app_def_v}, we have:
\begin{equation}
	\begin{split}
		\sum_{v\in V} w_{\epsilon}(v) |\<\psi|v\>|^2 &= \left( \sum_{v\in V} \||v^{*}\>\|^2 \right) \left( \sum_{v\in V} w_{\epsilon}(v)|\<\psi|v\>|^2 \right)\\
		&\geq \left( \sum_{v\in V} \||v^{*}\>\| \cdot \sqrt{w_{\epsilon}(v)} |\<\psi|v\>| \right)^2\\
		&= \left( \sum_{v\in V} \sqrt{w_{\epsilon}(v)} \cdot |\<\psi |v^{*}\>| \right)^2= | \<\psi| \sum_{i\in V} \sqrt{w_{\epsilon}(v)}|v^{*}\> |^2\\
		&= \| \sum_{v\in V} \sqrt{w_{\epsilon}(v)} |v^{*} \>\|^2\\
		&= \theta'_{\epsilon}(G_{\epsilon},\bm{w}_{\epsilon}).
	\end{split}
\end{equation}

In the first line, the equality holds because $\Tr Z = 1$, which implies that $\sum_{v\in V} \||v^*\>\|^2 = 1$. The second line follows from the Cauchy–Schwarz inequality. The equalities in the third and fourth lines arise from the definitions of $\{|v\>\}, v \in V$ and $|\psi\>$ in Eq.~\eqref{app_def_v}. The last line comes from the fact that $\theta'_{\epsilon}(G_{\epsilon},\bm{w}_{\epsilon}) = \sqrt{\bm{w}_{\epsilon}}^T Z \sqrt{\bm{w}_{\epsilon}}$, implying $\theta'_{\epsilon}(G_{\epsilon},\bm{w}_{\epsilon}) = \|\sum_{v\in V} \sqrt{w_{\epsilon}(v)} \cdot |v^{*}\> \|^2$.

Thus, the vectors $\{|v\>\}, v \in V$ and $|\psi\>$ defined in Eq.~\eqref{app_def_v} form a feasible, though not necessarily optimal, solution to $\theta_{\epsilon}(G_{\epsilon}, \bm{w}_{\epsilon})$. This concludes the proof of the lemma.
\end{proof}

	\section{$\epsilon$-Contextuality for a single Qubit system.}\label{proof_theo_qubit}
    \subsection{Kochen-Specker and Bell-Mermin hidden-variable models for qubits fail to be $\epsilon$-ONC}

In the Kochen-Specker model~\cite{kochen1990problem}, the hidden variable space is represented as a unit sphere, with a hidden variable $\lambda$ denoted by a unit vector $\vec{\lambda}$. The probability measure of the hidden variables is associated with the qubit state $|\psi\rangle$ being measured. Suppose the Bloch vector of this qubit state is $\vec{\psi}$; then the probability measure of $\lambda$ is defined as 
\begin{eqnarray}
    \text{Pr}_{\vec{\psi}}(\lambda) = 
    \begin{cases}
        \frac{1}{\pi} \vec{\psi} \cdot \vec{\lambda}, & \text{when } \vec{\psi} \cdot \vec{\lambda} > 0; \\
        0, & \text{when } \vec{\psi} \cdot \vec{\lambda} \leq 0.
    \end{cases}
\end{eqnarray}

Without loss of generality, we can assume $\vec{\psi} = (0,0,1)$ and we consider a projector $|\phi\rangle\langle\phi|$ that is ``almost''-orthogonal to $|\psi\rangle\langle\psi|$, represented by the Bloch vector $\vec{\phi} = (\sin{\theta_c}, 0, -\cos{\theta_c})$ for some critical angle $0 < \theta_c \ll \pi/2$. Hidden variables in the common region of the two hemispheres centered on $\vec{\psi}$ and $\vec{\phi}$ assign the same value in $\{0,1\}$ for these two ``almost''-orthogonal projectors, and the probability measure of these hidden variables is given by 
\begin{equation}
\begin{split}
     \frac{1}{2} \left( 1 - \int_{0}^{\frac{\pi}{2}-\theta_c} 2\pi \sin{\theta} \frac{\cos{\theta}}{\pi} d\theta \right) 
    = \frac{1}{2} \left( 1 - \int_{0}^{\frac{\pi}{2}-\theta_c} d\left(-\frac{1}{2}\cos{2\theta}\right) \right)
    = \frac{1}{4} - \frac{1}{4}\cos{2\theta_c}=\frac{1}{2} - \frac{1}{2}\cos^2{\theta_c}.
\end{split}
\end{equation}
This probability is greater than the overlap of these two ``almost''-orthogonal projectors, which is given by 
\begin{equation}
    \frac{1}{2} \left( 1 + \vec{\psi} \cdot \vec{\phi} \right) = \frac{1}{2} - \frac{1}{2}\cos{\theta_c},
\end{equation}
since $\cos{\theta_c} > \cos^2{\theta_c}$ for $0 < \theta_c \ll \pi/2$. 

Therefore, from the definition of $\epsilon$-ONC model, we see that for two projectors $v, w$ whose overlap is $\epsilon = \frac{1}{2} - \frac{1}{2}\cos{\theta_c}$, the probability 
$\text{Pr}\left[f_v^{C} = f_{w}^{C'}  \right] > \epsilon$ for $f_v^{C} : \Lambda \to \{0,1\}$. For a single qubit, each context consists of only two orthogonal projectors, therefore we have $\text{Pr}\left[f_v^{C} \neq f_{v}^{C'}  \right] > \epsilon$.

    Bell's hidden variable model for qubits was first proposed in~\cite{bell1966problem} and later adapted into a more intuitive form by Mermin in~\cite{mermin1993hidden}. In this model, the hidden variable space can be understood as a unit sphere, and a hidden variable $\lambda$ is represented by a unit vector $\vec{\lambda}$ on it, with a uniform probability distribution over the unit sphere given by $\text{Pr}(\lambda) = \frac{1}{4\pi}$. For a qubit state $|\psi\rangle$ with its Bloch vector denoted as $\vec{\psi}$, when implementing the projective measurement $|\phi\rangle\langle\phi|$ associated with the Bloch vector $\vec{\phi}$, the hidden variable $\lambda$ assigns the value 1 if $\vec{\phi} \cdot (\vec{\psi} + \vec{\lambda}) > 0$. Otherwise, the value 0 is assigned. 

Without loss of generality, we assume $\vec{\psi} = (0,0,1)$ and consider two ``almost''-orthogonal projectors $|\phi\rangle\langle\phi|$ and $|\phi'\rangle\langle\phi'|$ with Bloch vectors given by $\vec{\phi} = (\cos{\theta_c}, 0, \sin{\theta_c})$ and $\vec{\phi'} = (-\cos{\theta_c}, 0, \sin{\theta_c})$, for some critical angle $0 < \theta_c \ll \pi/2$. The probability measure of the hidden variables that assign the same value in $\{0,1\}$ to these two ``almost''-orthogonal projectors is given by 
\begin{equation}
    1 - 2\int_{0}^{\frac{\pi}{2} - \theta_c} 2\pi \sin\theta \frac{1}{4\pi} d\theta = 1 - \int_{0}^{\frac{\pi}{2} - \theta_c} d(-\cos\theta) = \sin{\theta_c}.
\end{equation}
Nevertheless, the overlap of these two ``almost''-orthogonal projectors is smaller, given by:
\begin{equation}
    \frac{1}{2} (1 + \vec{\phi} \cdot \vec{\phi'}) = \frac{1}{2} - \frac{1}{2}\cos{2\theta_c} = \sin^2{\theta_c}.
\end{equation}
Therefore, once again the model fails the $\epsilon$-ONC condition for $\epsilon = \sin^2{\theta_c}$.
    
    \subsection{Proof of Theorem 3}
For the convenience of the reader, we restate Theorem 3 from the main text as follows:
\begin{theorem}
		There exist non-contextuality inequalities and corresponding qubit correlations that certify quantum contextuality against $\epsilon$-ONC models for $\epsilon \in (0, \frac{1}{2})$, thereby exhibiting $\epsilon$-contextuality for a qubit.
\end{theorem}
	
	\begin{proof}
Consider a set of $2n$ ($n\geq 2$) qubit states defined as $S := \{|v_k^0\rangle, |v_k^1\rangle\}, k \in \{1, \ldots, n\}$, where
\begin{equation}\label{app_qubit_vec}
	|v_k^0\rangle = \cos\left(\frac{k\pi}{2n}\right)|0\rangle + \sin\left(\frac{k\pi}{2n}\right)|1\rangle,\qquad
	|v_k^1\rangle = -\sin\left(\frac{k\pi}{2n}\right)|0\rangle + \cos\left(\frac{k\pi}{2n}\right)|1\rangle.
\end{equation}
Clearly, $|v_k^0\rangle$ is orthogonal to $|v_k^1\rangle$ for all $k$. Furthermore, the overlap between $|v_k^0\rangle$ and $|v_{k+1}^1\rangle$, as well as between $|v_k^1\rangle$ and $|v_{k+1}^0\rangle$ for all $k=1,\ldots, n-1$, is $|\sin\left(\frac{\pi}{2n}\right)|$. We define the overlap of their corresponding projectors to be $\epsilon$, i.e., $\epsilon:=\sin^2\left(\frac{\pi}{2n}\right)$. Since $n \geq 2$, we have $\epsilon \leq \frac{1}{2}$. We refer to these vectors as $\epsilon$-orthogonal. Notably, $|v_n^0\rangle$ is $\epsilon$-orthogonal to $|v_1^0\rangle$, and $|v_n^1\rangle$ is $\epsilon$-orthogonal to $|v_1^1\rangle$.

\begin{figure}[t]
	\centering
	\subfigure[]{
		\begin{minipage}[t]{0.25\textwidth}
			\centering
			\includegraphics[width=1\textwidth]{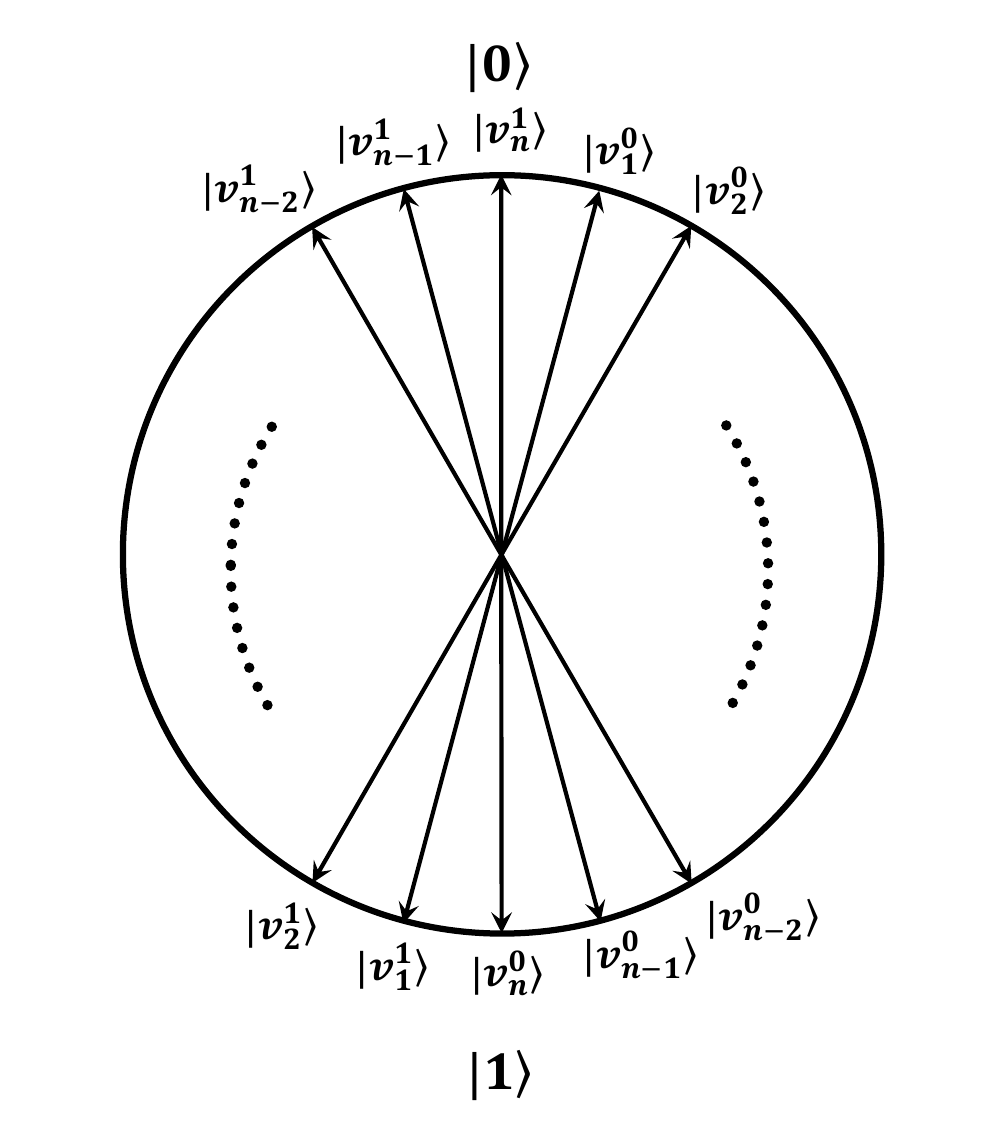}
		\end{minipage}
	}
	\qquad \qquad
	\subfigure[]{
		\begin{minipage}[t]{0.25\textwidth}
			\centering
			\includegraphics[width=1\textwidth]{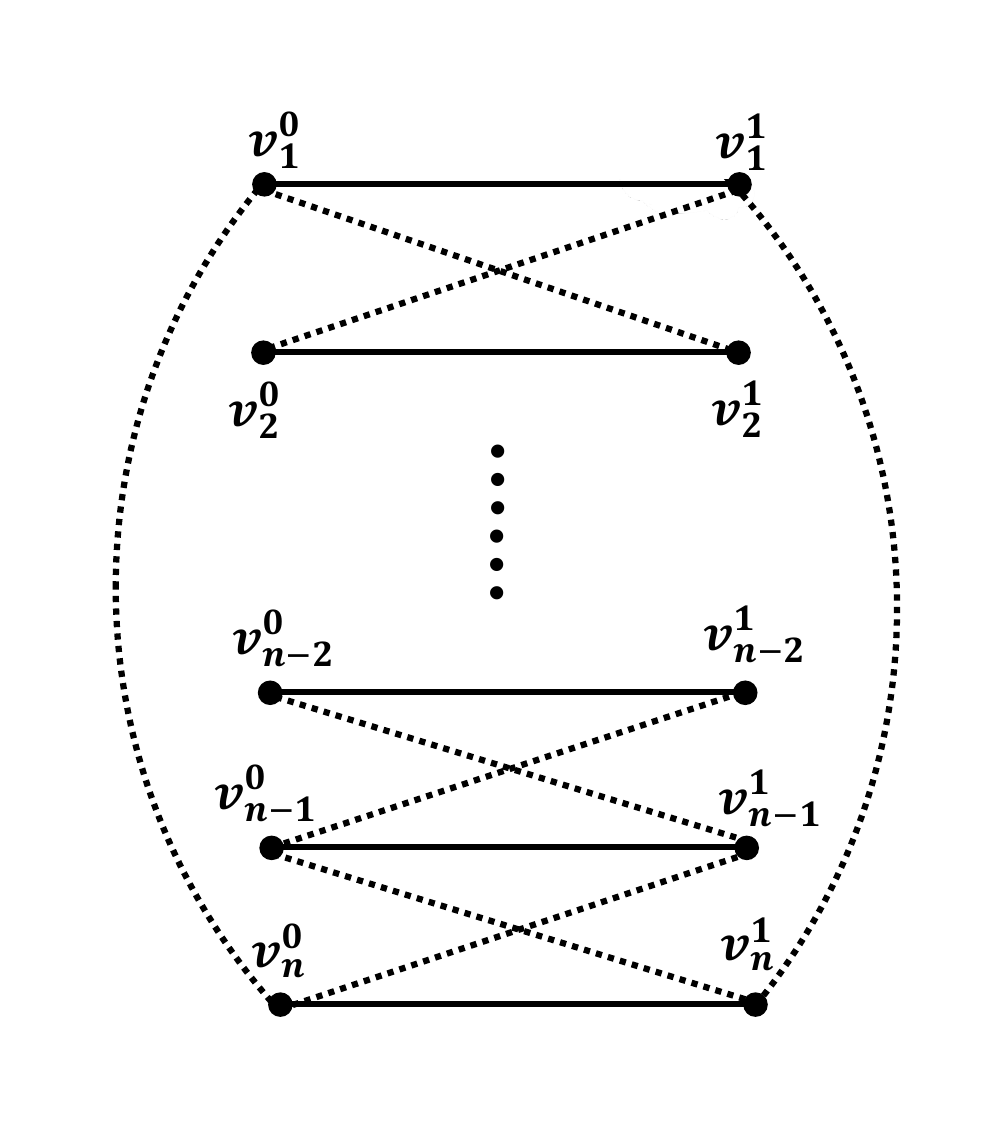}
		\end{minipage}
	}
	\caption{(a) The set of qubit states $S$ shown on the XZ plane of the Bloch sphere. (b) The $\epsilon$-orthogonality graph for $S$, where solid edges denote orthogonality and dotted lines denote $\epsilon$-orthogonality.}
	\label{qubit_cont}
\end{figure}

The $\epsilon$-orthogonality graph $G_{\epsilon}$ of these vectors is illustrated in Fig.~\ref{qubit_cont} which contains $n$ edges and $2n$ $\epsilon$-orthogonal edges. Each vertex has a weight of $1$. Since the vectors form $n$ complete bases in a two-dimensional Hilbert space, we have:
\begin{equation}
	\sum_{k=1}^n \left( |v_k^0\rangle\langle v_k^0| + |v_k^1\rangle\langle v_k^1| \right) = n \cdot \mathbb{I}_2,
\end{equation}
thus achieving the $\epsilon$-Lovász theta number $\theta_{\epsilon}(G_{\epsilon}) = n$ with respect to any qubit state $|\psi\>$.

Recall that the $\epsilon$-independence number of $G_{\epsilon}$ can be computed from two related graphs: $G'$, which includes only edges of $G_{\epsilon}$, and $G''$, which includes both edges and $\epsilon$-edges, treating $\epsilon$-edges as ``standard'' edges. It is straightforward to verify from Fig.~\ref{qubit_cont} that $\alpha(G') = n$ and $\alpha(G'') = n - 1$. Thus we have
\begin{equation}
	\alpha_{\epsilon}(G_{\epsilon}) = \epsilon \cdot \alpha(G') + (1 - \epsilon) \cdot \alpha(G'') = n - 1 + \epsilon.
\end{equation}

Therefore, for any $\epsilon \in (0, \frac{1}{2})$, there exists a finite set of $2n$ vectors defined as in Eq.~\eqref{app_qubit_vec} with $n \geq \lceil \frac{\pi}{2\arcsin(\sqrt{\epsilon})} \rceil$, exhibiting the relaxed version of state-independent qubit contextuality. For large $n$, we note that the parameter $\epsilon = \sin^2\left(\frac{\pi}{2n}\right)$ used to quantify the relaxation of the NCHV model becomes very small and at the same time the separation between quantum and classical ($\epsilon$-ONC) bounds increases, thus a single qubit system is arbitrarily close to being contextual. 
\end{proof}

 \section{Using the derived quantum correlations in an SDI protocol}\label{app_EAT}
 
The quantum correlations certifying $\log_2 d$ bits of randomness may be used in contextuality-based SDI randomness expansion protocols wherein the task is to expand a short fully random seed into a larger string of uniformly random bits. In these protocols, some initial randomness is required to perform the contextuality tests. To achieve a net positive generation rate, spot-checking protocols are commonly used. More precisely, an honest user runs the device for $N$ rounds, and each round is designated either as a {\em test round} with a low probability $\gamma \in (0, 1)$, or as a {\em randomness generation round} with a high probability $(1 - \gamma)$. In a test round, initial randomness is used to select an input (a context) to implement the contextuality test; and in a randomness generation round, a specific input setting $c^*$ (the context corresponding to a clique of size $d$ in the orthogonality graph $\mathcal{G}_d$) is chosen to generate raw randomness. At the end of the protocol, the quality of the raw randomness is assessed by analyzing the observed input-output data from the test rounds. Provided the protocol does not abort, the randomness accumulated during the randomness generation rounds is extracted by applying a suitable quantum-proof randomness extractor. We give a specific contextuality-based randomness expansion protocol below.

	\begin{protocol}{Contextuality-based SDI QRE}
	\label{fig:prot}
		\textbf{Parameters:}\\
		$N$: number of rounds.\\
		$\gamma\in(0,1)$: probability of a test round.\\
		$c^*$: distinguished context for generation rounds.\\
		$g$: a one-to-one function that maps outcomes of $c^*$ to $[d]$.\\
		$\omega_{\text{exp}}$: expected score of $\mathcal{I}_d$ for an honest implementation.\\
		$ \delta$: width of the statistical confidence interval for game score.\\
		$f_{\text{min}}$: min-trade off function.\\
		$\epsilon_{\text {ext }}\in(0,1)$: extractor error.\\
		$\epsilon_s \in(0,1)$: smoothing parameter.\\
		$\epsilon_{\text{EAT}} \in(0,1)$: entropy accumulation error. (Choose to  $\leq \epsilon_{\mathrm{ext}}+2\epsilon_s$)\\
		$R_{\text{ext}}$: quantum-proof $(k, \epsilon_{\mathrm{ext}}+2\epsilon_s)$-strong extractor.\\
		$l_{\text{ext}}$: entropy loss induced by $R_{\text{ext}}$.\\
		\hline
		\hline
		\textbf{Protocol}
		\begin{enumerate}
			\item [1.] For every round $i \in[N]$, choose $T_i \in\{0,1\}$ such that $T_i=1$ with probability $\gamma$.
			
			\textbf{If $T_i=1$:}
			
			Test round: Play a single round of $\mathcal{I}_d$ where the input $c$ (a context) is chosen uniformly at random using the initial private randomness. Record the input $c_i$ and the outputs $A_i$ of the sequential measurements in this context. Set the product of $A_i$ (outputs of the sequential measurements in this context) as the score $\omega_i$ of this round.
			
			\textbf{Else $T_i=0$:}
			
			Randomness generation round: Input $c^*$ into the device and record the outputs $A_i$ of the sequential measurements in this context. Set the score of this round as $\omega_i=0$. Record $g(A_i)$ as the raw randomness $r_i$.
			
			\item[2.] Compute the empirical score $\omega_{\text{obs}}:=\frac{\sum_{i=1}^N \omega_i}{\gamma \cdot N}.$
			
			\textbf{If} $\omega_{\text{exp}}-\delta<\omega_{\text{obs}}<\omega_{\text{exp}}+\delta$:
			
			Extraction: Apply a strong quantum-proof randomness extractor $R_{\text{ext}}$ to the raw randomness string $\bm{r}$ to extract $Nf_{min}(\omega_{\text{exp}}-\delta)-l_{\text{ext}}$ bits of $\epsilon_{\text{ext}}+2\epsilon_s$-close to uniform random outputs.

			\textbf{Else:}
			
			Abort: Abort the protocol.	
		\end{enumerate}
	\end{protocol}

The uncertainty of a variable $A$ given side information $B$, evaluated for a joint state $\rho_{AB}$, is lower bounded by the conditional $\epsilon$-smooth min-entropy~\cite{renner2004smooth}, denoted as $H_{min}^{\epsilon}(A|B)_{\rho_{AB}}$. The Entropy Accumulation Theorem (EAT)~\cite{arnon2018practical, dupuis2020entropy} is an information-theoretic tool that provides a way to lower bound the conditional $\epsilon$-smooth min-entropy accumulated over multiple rounds of a protocol, it allows the elevation of i.i.d. analyses to the non-i.i.d. setting. EAT achieves this by relating it to an average quantity that is straightforward to calculate, subject to some correction terms. 
For an $N$-round protocol such as shown below, under the condition that the protocol does not abort, the EAT states that 
\begin{equation}
	H_{min}^{\epsilon}(\bm{A}|\bm{X}E) \geq N f_{\text{min}}(\omega) - O(\sqrt{N}),
\end{equation}
where $\bm{A}$ represents the $N$-round outputs of the device, $\bm{X}$ denotes the $N$-round inputs, and $E$ is the adversary's side information. The function $f_{\text{min}}(\omega)$, known as the min-tradeoff function, is a convex, differentiable function that maps the game score $\omega$ to $\mathbb{R}$ and lower bounds the worst-case von Neumann entropy:
\begin{equation}
	f_{\text{min}}(\omega) \leq \inf H(A|XE),
\end{equation}
where the right-hand side represents the single-round conditional von Neumann entropy compatible with the game score $\omega$.

To apply this theorem, it is necessary to confirm that the protocol setting meets the EAT requirements. In the following, we review the definitions of EAT channels as presented in~\cite{arnon2018practical, dupuis2020entropy}, adapting them to the single-party setting required for our scenario, showing that the contextuality-based SDI randomness expansion protocol described in the main text meets the EAT requirements.

\begin{definition}(EAT Channels)\\The EAT Channels $\mathcal{M}_i: R_{i-1}\rightarrow R_i A_iX_iW_i$ for $i \in[N]$ are completely positive trace-preserving (CPTP) maps such that
	\begin{itemize}
		\item [1.]$A_iX_iW_i$ are finite-dimensional classical random variables ($A_i$ is measurement outcome, $X_i$ is measurement input, $W_i$ is a random variable evaluating the game score of the contextuality test), $R_i$ are arbitrary quantum registers (holding information about the quantum state at the $i$-th round).
		\item [2.] For any input state $\rho_{R_{i-1} R'}$ where $R'$ is a register isomorphic to $R_{i-1}$, the classical value $W_i$ (the game score of the contextuality test at the $i$-th round) can be measured from the marginal $\rho_{A_iX_i}$ (the classical random variables at the round) of the output state $\rho_{R_i A_iX_iW_i R'}=\left(\mathcal{M}_i \otimes \mathcal{I}_{R'}\right)\left(\rho_{R_{i-1} R'}\right)$ without changing the state.
		\item [3.] For any initial state $\rho_{R_0E}$, the final state $\rho_{\mathcal{A}^N\mathcal{X}^N\mathcal{W}^N E}=((\text{Tr}_{R_N}\circ \mathcal{M}_N \circ \cdots\mathcal{M}_1)\otimes \mathcal{I}_E)\rho_{R_0E}$ satisfies the Markov condition $\mathcal{A}^{i-1} \leftrightarrow \mathcal{X}^{i-1}E \leftrightarrow X_i$ for each $i \in[N]$. Here $\mathcal{A}^{i-1}=A_1, A_2, \ldots, A_{i-1}$ and similarly for the other random variables.
	\end{itemize}
\end{definition} 

The channels $\mathcal{M}_i$ describe the actions of the device, including the choice of input settings $X_i$ determined by some initial private randomness, the device output $A_i$, and the observed game score $W_i$. The first condition in the EAT channels, compatible with the protocol structure, specifies that the inputs $X_i$, outputs $A_i$, and observed game score $W_i$ are classical, while the arbitrary quantum register $R_i$ represents the quantum state stored by the device after the $i$-th round. The last condition reflects the sequential nature of the protocol, where the Markov chain condition implies that the device input $X_i$ in the $i$-th round is conditionally independent of previous outputs $\mathcal{A}^{i-1}$. This holds because the device inputs are derived from an initial private random seed. Together, the second and last conditions characterize the adversary, who is permitted to hold a purification $E$ of the initial state of the device, and the joint state evolves with the sequential interaction through the application of the sequence of EAT channels.

Now, we recall the efficient derivation of min-tradeoff functions via the dual program of semidefinite programs, following the approach outlined in~\cite{brown2019framework}. In this work, we employ the min-entropy $H_{\text{min}}(A|XE)$ to bound the von Neumann entropy $H(A|XE)$ from below. The min-entropy is calculated through the guessing probability as follows:
\begin{equation}
	H_{\text{min}}(A|XE) := -\log_2 p_{\text{guess}}(A|XE).
\end{equation}
The guessing probability represents the probability that an adversary, Eve, applies the best possible strategy on her system $E$ to correctly guess the outcome $A$ of a measurement setting $X=x^*$. Suppose the state shared between the device and Eve is $\rho_{AE}$, the guessing probability is defined as
\begin{equation}
	p_{\text{guess}}(A|X=x^*,E) := \max_{\bm{p}\in \mathrm{Q}} \sum_{a \in A} p(a, e = a | X = x^*, E) = \max_{\{P_e\}, \rho_{AE}} \sum_a \text{Tr}(P_a \otimes P_{e=a} \rho_{AE}),
\end{equation}
where $\{P_e\}$ denotes the projective measurement on Eve's system. In (semi-)device-independent scenarios, Eve can optimize her choice of measurement $\{P_e\}$ and the shared state $\rho_{AE}$ to maximize her guessing probability, covering all states and measurements that the device might use to achieve the best performance.

However, this optimization problem is challenging to efficiently solve because the objective function is nonlinear, and the optimal dimension of the quantum system is unknown, making the problem difficult to directly manage. We can relax this problem by employing the NPA (Navascués-Pironio-Acín) hierarchy~\cite{navascues2007bounding,navascues2008convergent}, which replaces the tensor product structure of operators with commutation relations, translating the optimization problem into a SDP that can be efficiently computed. The NPA hierarchy introduces a sequence of correlation sets $\mathrm{Q}_k$, forming an outer approximation to the quantum correlation set, where $\mathrm{Q}_1 \supsetneq \mathrm{Q}_2 \supsetneq \mathrm{Q}_3 \supsetneq \cdots \supsetneq \mathrm{Q}$. The solutions of these SDPs yield increasingly tighter upper bounds on the guessing probability.

The guessing probability must also be compatible with the observed game score $\omega$ for the contextuality test $\mathcal{I}_d$ (in our case, we use $\mathcal{I}_3$ presented in the main text). Thus, the relaxed guessing probability SDP can be expressed as
\begin{equation}
	\begin{split}
		{p^{(k)}_{\text{guess}}}(\omega) = \max_{\bm{p}} & \quad \sum_{a \in A} p(a, e = a | X = x^*, E) \\
		\text{s.t.} & \quad \bm{p} \in \mathrm{Q}_k, \\
		& \quad \mathcal{I}(\bm{p}) = \omega.
	\end{split}
\end{equation}
This program can be expressed in a standard form:
\begin{equation}
	\begin{split}
		{p^{(k)}_{\text{guess}}}(\omega) = \max_P & \quad \text{Tr}[G P] \\
		\text{s.t.} & \quad \text{Tr} \left[S_i P\right] = s_i \quad \text{for all } i, \\
		& \quad \text{Tr} \left[I P\right] = \omega, \\
		& \quad P \geq 0,
	\end{split}
\end{equation}
where each matrix $S_i$ and value $s_i$ corresponds to constraints on the NPA moment matrix $P$, such as commutation conditions for operators, the non-negative, normalization and no-signaling conditions for corresponding probabilities.
The dual program is given by
\begin{equation}
	\begin{split}
		{d^{(k)}_{\text{guess}}}(\omega) = \min_{\bm{\lambda}} & \quad \lambda_0 \omega + \sum_i \lambda_i s_i \\
		\text{s.t.} & \quad \lambda_0 I + \sum_i \lambda_i S_i \geq G.
	\end{split}
\end{equation}
According to the weak duality theorem, the solution to the dual program provides an upper bound on the solution of the primal SDP. The optimal solution for the Lagrange multipliers $\hat{\bm{\lambda}}$ at any fixed game score $\hat{\omega}$ forms an affine function $g_{\hat{\omega}}(\omega)$ that bounds the guessing probability curve ${p^{(k)}_{\text{guess}}}(\omega)$ from above for any $\omega$. Assuming strong duality holds, i.e., ${d^{(k)}_{\text{guess}}}(\hat{\omega}) = {p^{(k)}_{\text{guess}}}(\hat{\omega})$, then $g_{\hat{\omega}}(\omega)$ is tangent to the guessing probability curve ${p^{(k)}_{\text{guess}}}(\omega)$ at $\hat{\omega}$. This can be seen by observing that for any feasible solution $P$ of the primal program with game score $\omega$, we have:
\begin{equation}
	\begin{split}
		0 &\leq \text{Tr}\left[(\hat{\lambda}_0 I + \sum_i \hat{\lambda}_i S_i - G)P\right] \\
		&= \hat{\lambda}_0 \text{Tr}[I P] + \sum_i \hat{\lambda}_i \text{Tr}[S_i P] - \text{Tr}[G P] \\
		&= \hat{\lambda}_0 \omega + \sum_i \hat{\lambda}_i s_i - \text{Tr}[G P] \\
		&= g_{\hat{\omega}}(\omega) - \text{Tr}[G P].
	\end{split}
	\end{equation}
	Thus, from the optimal solution of the dual program, we obtain a family of affine functions $\{g_{\hat{\omega}}(\omega)\}_{\hat{\omega}}$ that upper bound the guessing probability curve. The family of functions $\{f_{\hat{\omega}}(\omega) := -\log_2\left(g_{\hat{\omega}}(\omega)\right)\}_{\hat{\omega}}$ provides a lower bound on the min-entropy $H_{\text{min}}(A|X = x^*, E)$ (and hence also on the von Neumann entropy $H(A|X = x^*, E)$) for each game score $\omega$. Notably, since $g_{\hat{\omega}}(\omega) > 0$ for all $\hat{\omega}$ as it upper bounds the guessing probability, the family of functions $\{f_{\hat{\omega}}(\omega)\}_{\hat{\omega}}$ are convex and differentiable, making them suitable candidates for the min-tradeoff function in the EAT theorem.

	We numerically calculated the min-entropy curve with respect to the game score for $\mathcal{I}_3$ and generated a family of min-tradeoff functions accordingly, plotting them as below. Similar curves (and even possible enhancements achieved by computing the conditional von Neumann entropy) will be derived in forthcoming experimental implementations of contextuality-based SDI protocols. 
	
	\begin{figure}[H]
		\centering
		\includegraphics[width=0.45\linewidth]{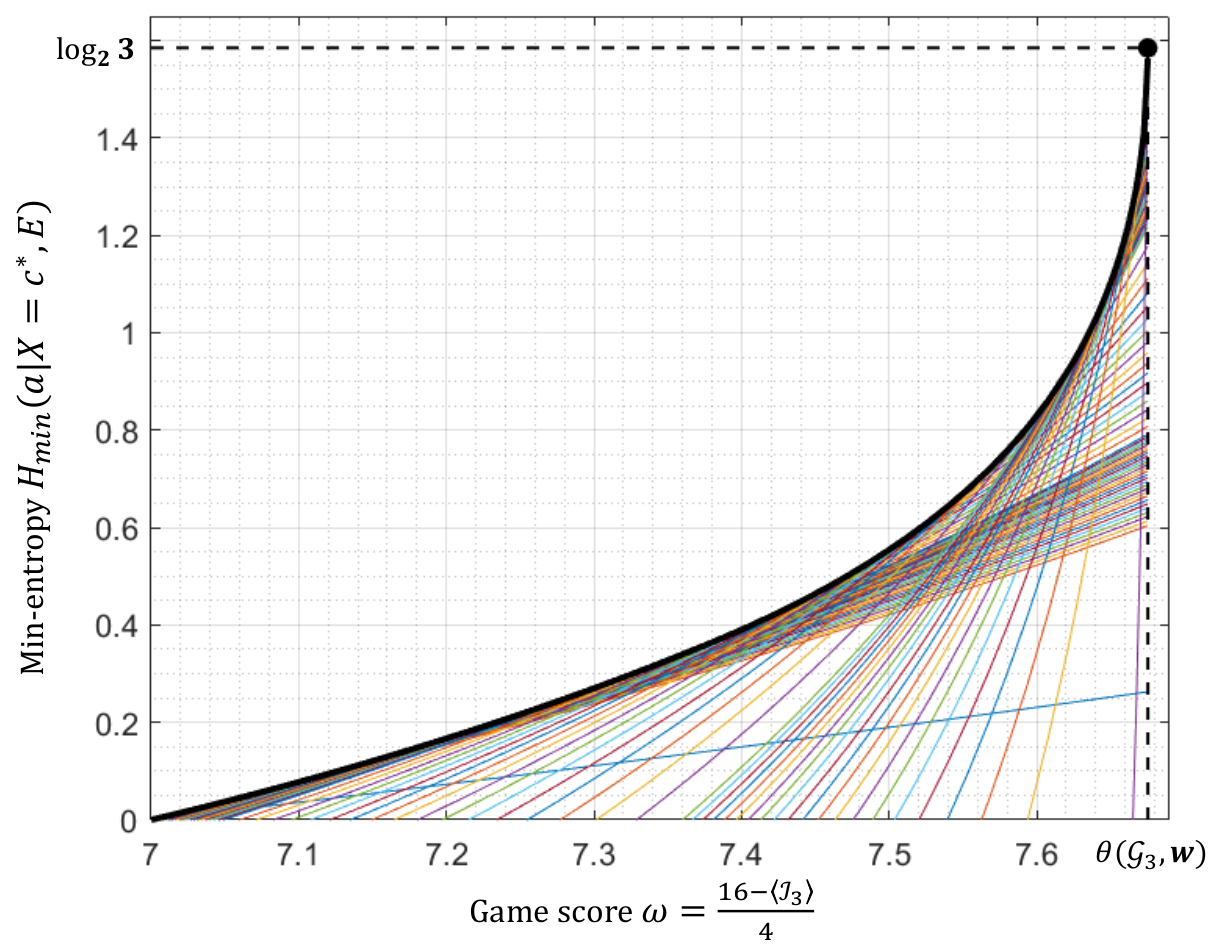}
		\caption{The min-entropy $H_{\text{min}}(a | X = c^*, E)$ versus the game score $\omega := \frac{(5d+1) - \<\mathcal{I}_d\>}{4}$ curve for $d=3$ with a family of min-tradeoff functions.}
		\label{min_entropy}
	\end{figure}

\section{Attacks on SDI protocols for certain contextuality tests.}	\label{SDI-attacks}
\subsection{Attacks on state-independent contextuality-based protocols by a quantum adversary}
State-independent contextuality (SI-C) as the name implies, allows for a contextuality test that does not depend on the choice of a particular quantum state. An SI-C can be written as a non-contextuality inequality for which the quantum value can be achieved by a set of observables $\{A_i\}_i$ acting on a fixed finite-dimensional Hilbert space with any state $\rho$ in the underlying space. Any KS set leads to a non-contextuality inequality of this type. Furthermore, it is also known that there are SI-C sets that do not correspond to a KS set~\cite{yu2012state}.

It is reasonably straightforward to observe that SI-C-based randomness certification protocols are not secure against a quantum adversary, even when the quantum values of these inequalities are observed. This is because the quantum value for a state-independent non-contextuality inequality in a Hilbert space of dimension $d$ can even be achieved by a maximally mixed state $\mathbb{I}/d$. When a quantum adversary shares the bipartite maximally entangled state $|\psi\>_{D,E} = \frac{1}{\sqrt{d}}\sum_{i=0}^{d-1}|i,i\>$ with the device operating the randomness certification protocol, the reduced system of the device can still observe the quantum value of the state-independent non-contextuality inequality. Furthermore, such an adversary is able to steer the state of the device depending on the publicly announced measurements made on the device, and thereby perfectly guess the corresponding measurement outcomes.


More specifically, in the contextuality test, a context $C$ includes $2 \leq k \leq d$ mutually commuting observables $A_1, \dots, A_k$ (with $A_i := \mathbb{I} - 2P_i$, where $P_i$ is a projector). By sharing a maximally entangled state, and the adversary measuring the same context with the device, the adversary's outcome for each observable $A_i$ in this context is fully correlated with the outcome of the device:
\begin{equation}
\begin{split}
    \<A_i \otimes A_i\>_{D,E} &= \<\psi|_{D,E} (\mathbb{I} - 2 P_i) \otimes (\mathbb{I} - 2 P_i) |\psi\>_{D,E} \\
    &= \<\psi|_{D,E} \left(\mathbb{I} \otimes \mathbb{I} - 2 P_i \otimes \mathbb{I} - 2 \mathbb{I} \otimes P_i + 4 P_i \otimes P_i \right) |\psi\>_{D,E} \\
    &= 1 + \<\psi|_{D,E} \left( -2 P_i \otimes (P_i + P_i^{\perp}) - 2 (P_i + P_i^{\perp}) \otimes P_i + 4 P_i \otimes P_i \right) |\psi\>_{D,E} \\
    &= 1 - 2 \<\psi|_{D,E} \left( P_i \otimes P_i^{\perp} + P_i^{\perp} \otimes P_i \right) |\psi\>_{D,E} = 1.
\end{split}
\end{equation}
To circumvent this attack, one needs to use a private random source to choose the input settings, wherein the measurement settings used are not publicly known to the adversary. Alternatively, the honest user needs to perform several tests incorporating the full input-output statistics of the device. 

\subsection{Attacks on magic arrangement-based protocols by  hybrid classical adversaries}
In this subsection, we present attacks on a specific class of protocols based on a type of contextuality test known as ``magic arrangements''. These contextuality tests include well-known examples such as the Peres-Mermin magic square and the Mermin magic pentagram. We consider attacks by hybrid classical adversaries~\cite{bierhorst2018experimentally} who do not share any quantum or post-quantum resources with the device. However, they are allowed to prepare non-disturbing (consistent) correlations in advance and program them into the device.

A particular class of quantum contextuality proofs is explored through the classical and quantum realizability of \textit{signed arrangements}~\cite{arkhipov2012extending}. A signed arrangement is defined as a finite connected hypergraph $G = (V, E)$ in which each vertex belongs to exactly two hyperedges, with hyperedge labeling $l: E \to \{+1, -1\}$ and a signed arrangement is thus represented by $A = (V, E, l)$. 

A classical realization of a signed arrangement $A = (V, E, l)$ is a map $f_c: V \to \{+1, -1\}$ such that $\prod_{v \in e} f_c(v) = l(e)$ for all $e \in E$. In contrast, a quantum realization of a signed arrangement $A = (V, E, l)$ involves a map $f_q$ that assigns observables with $\pm 1$ eigenvalues to vertices within a fixed finite-dimensional Hilbert space, where observables commute if their corresponding vertices belong to a common hyperedge, and $\prod_{v \in e} f_q(v) = l(e) \mathbb{I}$ holds for any $e \in E$.
\textit{Magic arrangements} are those that are not classically realizable but are quantumly realizable. They are characterized by the property that $\prod_{e \in E} l(e) = -1$~\cite{arkhipov2012extending}. 

In this subsection, we focus on non-disturbing (consistent) realizations of these signed arrangements. For any hyperedge $e \in E$, let $n_e$ denote the number of vertices in $e$, and let $\bm{o}_e$ represent the variable corresponding to the $\{+1,-1\}$ assignments of these $n_e$ vertices. Thus, $\bm{o}_e \in \{+1,-1\}^{n_e}$, and $o_e(v)$, for $v \in e$, is the assignment of $v$ in $\bm{o}_e$. For each $e \in E$, define the set $O_e := \{\bm{o}_e | \prod_{v \in e} o_e(v) = l(e)\}$.

The non-disturbing realization of a signed arrangement $A = (V, E, l)$ is represented by a behavior $\{p(\bm{o}_e | e)\}$ that satisfies the condition
\begin{equation}\label{app_nd_con}
    \sum_{\bm{o}_e \in O_e} p(\bm{o}_e | e) = 1, \quad \forall e \in E.
\end{equation}
In addition, the behavior $\{p(\bm{o}_e | e)\}$ should satisfy the non-disturbing conditions, that the marginal probability of the outcome for each vertex is independent of the hyperedge it belongs to. This implies that for all $v \in V$ and for all hyperedges $e, e' \in E$ such that $v \in e, e'$, we have:
\begin{equation}\label{app_nd_con2}
    \begin{split}
        p(o_e(v) = +1 | e) &= p(o_{e'}(v) = +1 | e'),\\
        p(o_e(v) = -1 | e) &= p(o_{e'}(v) = -1 | e').
    \end{split}
\end{equation}

We now provide concrete non-disturbing realizations for two signed arrangements illustrated in Fig.~\ref{arrangement_graph}. In graph theory terms, Fig.~\ref{k22} represents the graph $K_{2,2}$, which contains four vertices $v_1, v_2, v_3, v_4$ and four hyperedges $e_1, e_2, e_3, e_4$. Fig.~\ref{k3} represents the graph $K_3$, which contains three vertices $u_1, u_2, u_3$ and three hyperedges $f_1, f_2, f_3$. The labels of hyperedges $e_4$ and $f_3$ are $-1$, while the others are $+1$. A non-disturbing realization for Fig.~\ref{k22} is given by $\{p_{K_{2,2}}(\bm{o}_e | e)\}$, which is expressed as follows:
\begin{equation}
    \begin{split}
        p_{K_{2,2}}(o_{e_1}(v_1) = o_{e_1}(v_2) = +1 | e_1) &= p_{K_{2,2}}(o_{e_1}(v_1) = o_{e_1}(v_2) = -1 | e_1) = \frac{1}{2}, \\
        p_{K_{2,2}}(o_{e_2}(v_3) = o_{e_2}(v_4) = +1 | e_2) &= p_{K_{2,2}}(o_{e_2}(v_3) = o_{e_2}(v_4) = -1 | e_2) = \frac{1}{2}, \\
        p_{K_{2,2}}(o_{e_3}(v_1) = o_{e_3}(v_3) = +1 | e_3) &= p_{K_{2,2}}(o_{e_3}(v_1) = o_{e_3}(v_3) = -1 | e_3) = \frac{1}{2}, \\
        p_{K_{2,2}}(o_{e_4}(v_2) = +1, o_{e_4}(v_4) = -1 | e_4) &= p_{K_{2,2}}(o_{e_4}(v_2) = -1, o_{e_4}(v_4) = +1 | e_4) = \frac{1}{2},
    \end{split}
\end{equation}
with all other conditional probabilities being zero. The marginal probability of any outcome for each vertex is \( \frac{1}{2} \), i.e., \( p_{K_{2,2}}(o_e(v) | e) = \frac{1}{2}, \forall v, \forall e \text{ such that } v \in e \), thereby ensuring that the non-disturbing conditions hold. It is also straightforward to verify that for each hyperedge, the product of the outcomes of the vertices belonging to that hyperedge is always equal to the labeling of that hyperedge, thus satisfying the condition in Eq.~\eqref{app_nd_con}.
\begin{figure}[H]
    \centering
    \subfigure[]{
        \begin{minipage}[t]{0.3\textwidth}
            \centering
            \includegraphics[width=0.6\textwidth]{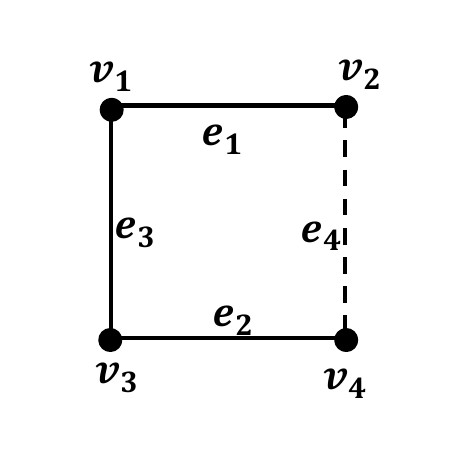}
            \label{k22}
        \end{minipage}
    }
    \subfigure[]{
        \begin{minipage}[t]{0.3\textwidth}
            \centering
            \includegraphics[width=0.6\textwidth]{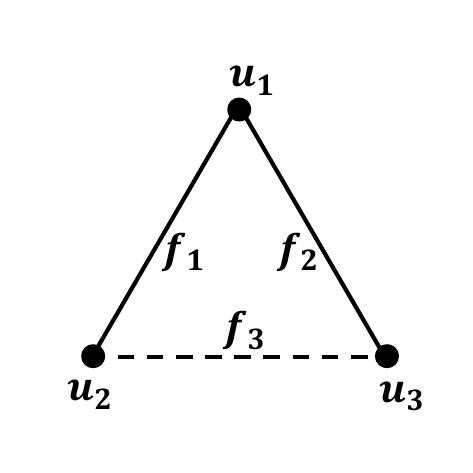}
            \label{k3}
        \end{minipage}
    }
    \caption{Two signed arrangements. (a) $K_{2,2}$ contains four vertices and four hyperedges, with the sign of $e_4$ being $-1$ and the others being $+1$. (b) $K_3$ contains three vertices and three hyperedges, with the sign of $f_3$ being $-1$ and the others being $+1$.}
    \label{arrangement_graph}
\end{figure}
Similarly, a non-disturbing realization for Fig.~\ref{k3} is given by $\{p_{K_3}(\bm{o}_e | e)\}$, which satisfies:
\begin{equation}
    \begin{split}
        p_{K_3}(o_{f_1}(u_1) = o_{f_1}(u_2) = +1 | f_1) &= p_{K_3}(o_{f_1}(u_1) = o_{f_1}(u_2) = -1 | f_1) = \frac{1}{2}, \\
        p_{K_3}(o_{f_2}(u_1) = o_{f_2}(u_3) = +1 | f_2) &= p_{K_3}(o_{f_2}(u_1) = o_{f_2}(u_3) = -1 | f_2) = \frac{1}{2}, \\
        p_{K_3}(o_{f_3}(u_2) = +1, o_{f_3}(u_3) = -1 | f_3) &= p_{K_3}(o_{f_3}(u_2) = -1, o_{f_3}(u_3) = +1 | f_3) = \frac{1}{2},
    \end{split}
\end{equation}
with all other conditional probabilities being zero.

To substantiate the main claim of this subsection, we need to introduce two terminologies from graph theory.

\begin{definition}
    The intersection graph of a signed arrangement $A = (V, E, l)$ is defined as a graph $H = (\widetilde{V}, \widetilde{E})$, where $\widetilde{V} = E$ and $\widetilde{E} = V$. An edge $v \in \widetilde{E}$ exists between vertices $e_1, e_2 \in \widetilde{V}$ in $H$ if $v \in e_1 \cap e_2$ in $A$. The labeling $l$ of hyperedges in $A$ is automatically translated into the labeling of vertices in $H$, denote it by $\widetilde{l}$.
\end{definition}
The famous Peres-Mermin magic square and Mermin magic pentagram contextuality proofs can be interpreted as signed arrangements, as shown in Figures~\ref{MSquare} and~\ref{Mpentagram}, along with their corresponding intersection graphs.

\begin{figure}[H]
    \centering
    \subfigure[]{
        \begin{minipage}[t]{0.45\textwidth}
            \centering
            \includegraphics[width=\textwidth]{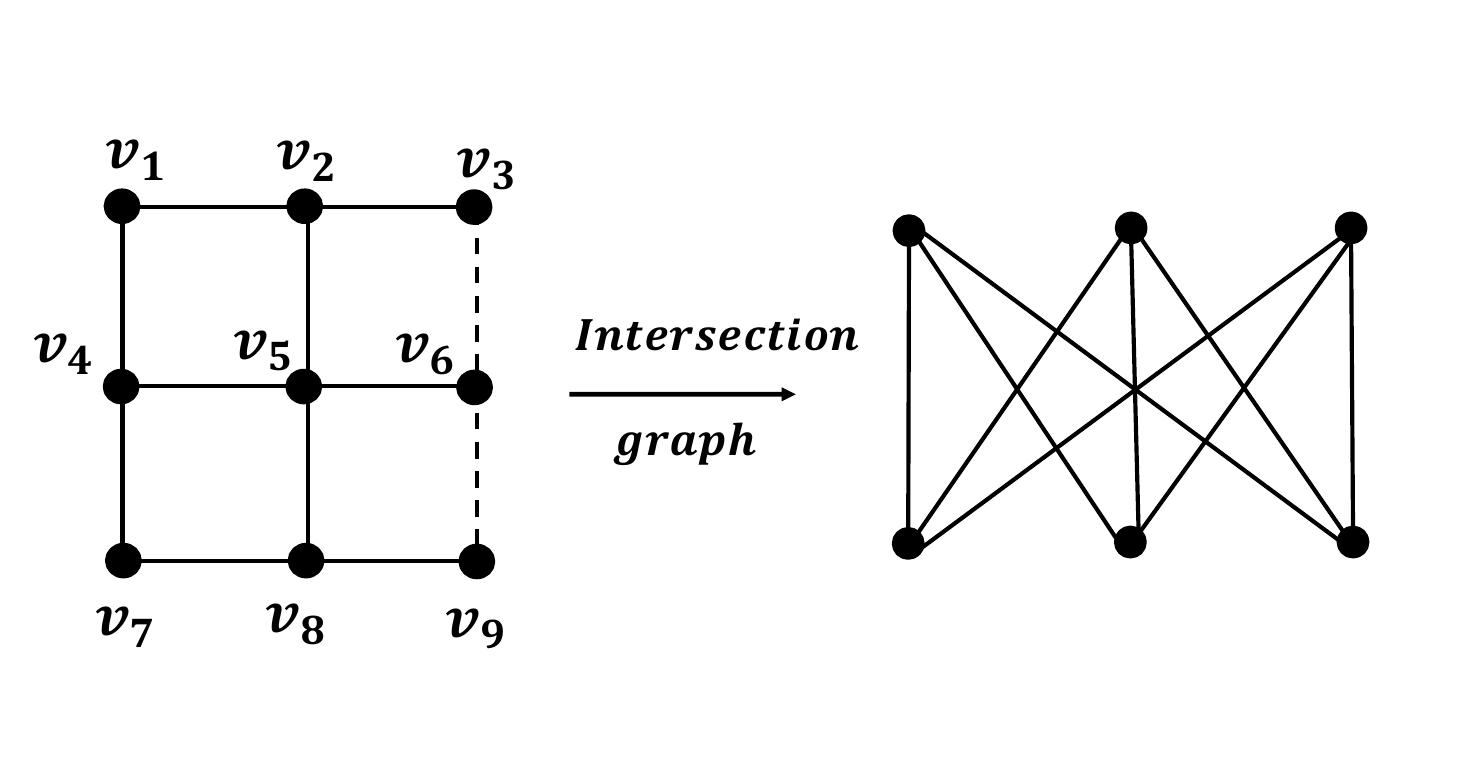}
            \label{MSquare}
        \end{minipage}
    }
    \subfigure[]{
        \begin{minipage}[t]{0.45\textwidth}
            \centering
            \includegraphics[width=\textwidth]{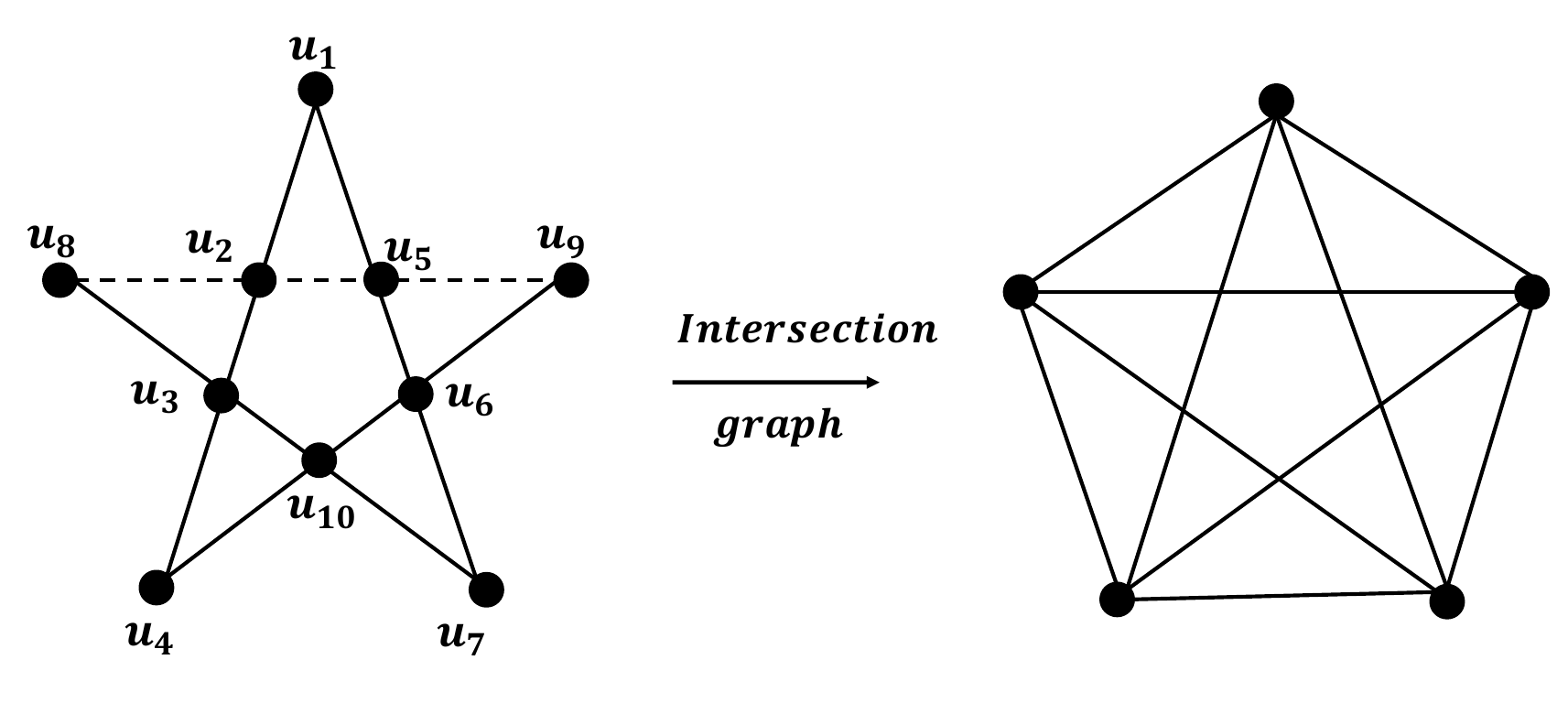}
            \label{Mpentagram}
        \end{minipage}
    }
    \caption{(a) The signed arrangement representation of the magic square and its corresponding intersection graph. (b) The signed arrangement representation of the magic pentagram and its corresponding intersection graph. The hyperedge labeling for the dashed hyperedge is $-1$, while the others are labeled $+1$.}
\end{figure}

\begin{definition}\label{def_topo_minor}
    A graph $H$ is considered a topological minor of $G$ if there exists an embedding of $H$ in $G$, represented by an injective map $\phi$ that maps each vertex $v$ of $H$ to a vertex $\phi(v)$ of $G$, and each edge $(u, v)$ of $H$ to a simple path from $\phi(u)$ to $\phi(v)$ in $G$, where these paths must be disjoint except at their endpoints.
\end{definition}

Now, we proceed to prove two useful propositions.

\begin{proposition}\label{app_prop1}
    The no-signaling realizability of a signed arrangement $A = (V, E, l)$ depends on $l$ only via its parity $\prod_{e \in E} l(e)$.
\end{proposition}

Arkhipov proved this proposition for classical and quantum realizability in~\cite{arkhipov2012extending}. Here, we provide a proof for non-disturbing realizability, motivated by his result.

\begin{proof}
    To prove this proposition, we show that the non-disturbing realizability of two signed arrangements $A = (V, E, l)$ and $A' = (V, E, l')$ are the same when $\prod_{e \in E} l(e) = \prod_{e \in E} l'(e)$. Specifically, suppose a non-disturbing realization for $A$ is $\{p(\bm{o}_e | e)\}$. We can construct a non-disturbing realization for $A'$ based on $\{p(\bm{o}_e | e)\}$.

    First, assume that $l'$ is obtained by flipping the signs assigned by $l$ to two hyperedges $a$ and $b$ that share a common vertex $v^*$. We can construct a non-disturbing realization $\{p'(\bm{o}_e | e)\}$ for $A'$ as follows:
    \begin{equation}\label{ns_box}
    p'(\bm{o}_e | e) :=
        \begin{cases}
            p(\bm{o}_e \setminus o_e(v^*), -o_e(v^*) | e), & \text{when } e = a \text{ or } b; \\
            p(\bm{o}_e | e), & \text{otherwise.}
        \end{cases}
    \end{equation}
    Here, $-o_e(v^*)$ denotes the negation of the assignment of $o_e(v^*)$. Thus, for each $e \in E$ defining the set $O'_e := \{\bm{o}_e | \prod_{v \in e} o_e(v) = l'(e)\}$, we have:
    \begin{equation}
   \sum_{\bm{o}_e \in O'_e} p'(\bm{o}_e | e) = \sum_{\bm{o}_e \in O_e} p(\bm{o}_e | e) = 1, \quad \forall e \in E.
    \end{equation}
    Additionally, the non-disturbing conditions hold for $\{p'(\bm{o}_e | e)\}$. Specifically, we have:
    \begin{equation}
    \begin{split}
        p'(o_e(v^*) | e) &= p(-o_e(v^*) | e), \quad \text{when } e = a \text{ or } b; \\
        p'(o_e(v) | e) &= p(o_e(v) | e), \quad\quad \; \forall v \in V \setminus v^*, \forall e \in E \text{ such that } v \in e.
    \end{split}
    \end{equation}

    Now, suppose that $l'$ is obtained by flipping the signs assigned by $l$ to two hyperedges $a$ and $b$ that do not share any common vertex. Ignoring the labeling function $l$, an arrangement is simply a finite connected hypergraph; hence, there must be a path of distinct edges $e_0, e_1, \ldots, e_n$, starting at $e_0 = a$ and ending at $e_n = b$, where any pair of adjacent hyperedges $e_i$ and $e_{i+1}$ ($i \in \{0, \ldots, n-1\}$) intersects at a vertex $v_i$. For each $v_i$, updating the non-disturbing realization $\{p(\bm{o}_e | e)\}$ as in Eq.~\eqref{ns_box} results in the desired non-disturbing realization $\{p'(\bm{o}_e | e)\}$ for $A'$. By repeatedly updating the non-disturbing realizations, one can transition from one labeling to any other labeling of equal parity.
\end{proof}

According to this proposition~\ref{app_prop1}, when discussing an arrangement we can ignore the concrete form of the labelling function $l$ and only consider its parity.

\begin{proposition}\label{app_prop2}
    Suppose the signed arrangements $A=(V,E,l)$ and $A'=(V',E',l')$ have the same parity of their labeling functions, and the intersection graph of $A$ is a topological minor of the intersection graph of $A'$. Then, if $A$ has a non-disturbing realization, it follows that $A'$ also has a non-disturbing realization.
\end{proposition}

\begin{proof}
    We prove that from the non-disturbing realization $\{p(\bm{o}_e | e)\}$ of $A$, we can construct a non-disturbing realization for $A'$.

    Denote the intersection graph of $A$ by $H = (\widetilde{V}, \widetilde{E})$ and the intersection graph of $A'$ by $G = (\widetilde{V}', \widetilde{E}')$. According to Definition~\ref{def_topo_minor}, there is an injective map $\phi$ that maps each vertex $\widetilde{v} \in \widetilde{V}$ to a vertex $\phi(\widetilde{v}) \in \widetilde{V}'$ of $G$. According to Proposition~\ref{app_prop1}, without loss of generality, we assume the vertex labeling function of $G$, $\widetilde{l}'$ (associated with the hyperedge labeling function $l'$ of $A'$), satisfies the following conditions: (1) for each vertex $\widetilde{v}$ of $H$, $\widetilde{l}'(\phi(\widetilde{v})) = \widetilde{l}(\widetilde{v})$; and (2) $\widetilde{l}'$ assigns the label $+1$ to the other vertices in $G$. 
    Furthermore, $\phi$ maps edges in $\widetilde{E}'$ to disjoint paths. We denote the path from $\phi(\widetilde{u})$ to $\phi(\widetilde{v})$ as $\text{PATH}_{\widetilde{u},\widetilde{v}}$, which is a subset of $\widetilde{V}'$.

    Due to the terminology of intersection graphs, the $\{+1, -1\}$-assignment of vertices in arrangement $A$ is translated to the $\{+1, -1\}$-assignment of edges in its intersection graph $H$. Furthermore, the assignment string $\bm{o}_e$ for vertices belonging to hyperedges in $e \in E$ is translated to the assignment string for edges with one endpoint being $\widetilde{v} = e$ in $\widetilde{V}$; we denote the number of these edges by $n_{\widetilde{v}}$ and the string by $\bm{\widetilde{o}}_{\widetilde{v}}$. The non-disturbing realization $\{p(\bm{o}_e | e)\}$ of $A$ can now be understood as $\{p(\bm{\widetilde{o}}_{\widetilde{v}} | \widetilde{v})\}$ for $H$. Additionally, we denote $\widetilde{O}_{\widetilde{v}} = \{\bm{\widetilde{o}}_{\widetilde{v}} \mid \prod_{\widetilde{e} \in \widetilde{E} : \widetilde{v} \in \widetilde{e}} \widetilde{o}_{\widetilde{v}} (\widetilde{e}) = \widetilde{l}(\widetilde{v})\}$. The condition in Equation~\eqref{app_nd_con} is translated to 
    \begin{equation}\label{app_con_nd_3}
        \sum_{\bm{\widetilde{o}}_{\widetilde{v}} \in \widetilde{O}_{\widetilde{v}}} p(\bm{\widetilde{o}}_{\widetilde{v}} | \widetilde{v}) = 1, \quad \forall \widetilde{v} \in \widetilde{V}.
    \end{equation}
    Moreover, the non-disturbing condition in Equation~\eqref{app_nd_con2} is translated for any $\widetilde{e} \in \widetilde{E}$ and for all $\widetilde{v}, \widetilde{v}' \in \widetilde{e}$:

    \begin{equation}\label{app_con_nd_4}
        \begin{split}
            p(\widetilde{o}_{\widetilde{v}} (\widetilde{e}) = +1 | \widetilde{v}) &= p(\widetilde{o}_{\widetilde{v}'} (\widetilde{e}) = +1 | \widetilde{v}'),\\
            p(\widetilde{o}_{\widetilde{v}} (\widetilde{e}) = -1 | \widetilde{v}) &= p(\widetilde{o}_{\widetilde{v}'} (\widetilde{e}) = -1 | \widetilde{v}').
        \end{split}
    \end{equation}

    Now, we construct the non-disturbing realization for $A'$ using these intersection graph terms as $\{p'(\bm{\widetilde{o}}_{\widetilde{v}'} | \widetilde{v}')\}$, such that (1) the outcomes assigned to edges of $\text{PATH}_{\widetilde{u}, \widetilde{v}}$ are \textbf{correlated} with the outcomes assigned to the edge $\widetilde{e} = (\widetilde{u}, \widetilde{v})$ with probabilities consistent with the behavior $\{p(\bm{\widetilde{o}}_{\widetilde{v}} | \widetilde{v})\}$, and (2) always assign output $+1$ to other edges in $\widetilde{E}'$. 

    In other words, for any $\widetilde{v}' \in \widetilde{V}'$ such that $\widetilde{v}' = \phi(\widetilde{v})$ for some $\widetilde{v} \in \widetilde{V}$, we have
    \begin{equation}
        p'(\bm{\widetilde{o}}_{\widetilde{v}}, \underbrace{+1, +1, \ldots, +1}_{n_{\widetilde{v}'} - n_{\widetilde{v}}} | \widetilde{v}') = p(\bm{\widetilde{o}}_{\widetilde{v}} | \widetilde{v}).
    \end{equation}
    For any $\widetilde{v}' \in \widetilde{V}'$ that belongs to some $\text{PATH}_{\widetilde{u}, \widetilde{v}}$ defined by $\phi$ for some $\widetilde{e} = (\widetilde{u}, \widetilde{v})$, but is not the endpoint (i.e., $\widetilde{v}' \neq \phi(\widetilde{v})$ and $\widetilde{v}' \neq \phi(\widetilde{u})$), we have: (Note that these paths defined by $\phi$ are disjoint except at their endpoints.)
    \begin{equation}
        p'({\widetilde{o}}_{\widetilde{v}}(\widetilde{e}), {\widetilde{o}}_{\widetilde{v}}(\widetilde{e}), \underbrace{+1, +1, \ldots, +1}_{n_{\widetilde{v}'} - 2} | \widetilde{v}') = p({\widetilde{o}}_{\widetilde{v}}(\widetilde{e}) | \widetilde{v}).
    \end{equation}
    Lastly, for the vertices $\widetilde{v}' \in \widetilde{V}'$ that do not belong to any path $\text{PATH}_{\widetilde{u}, \widetilde{v}}$ defined by $\phi$, we have:
    \begin{equation}
        p'(\underbrace{+1, +1, \ldots, +1}_{n_{\widetilde{v}'}} | \widetilde{v}') = 1.
    \end{equation}

    Now, we verify that $\{p'(\bm{\widetilde{o}}_{\widetilde{v}'} | \widetilde{v}')\}$ is a valid non-disturbing realization, i.e., that the conditions in Equations~\eqref{app_con_nd_3} and \eqref{app_con_nd_4} hold for $G$. For any $\widetilde{v}' \in \widetilde{V}'$, denote $\widetilde{O}'_{\widetilde{v}'} = \{\bm{\widetilde{o}}_{\widetilde{v}'} \mid \prod_{\widetilde{e}' \in \widetilde{E}' : \widetilde{v}' \in \widetilde{e}'} \widetilde{o}_{\widetilde{v}'} (\widetilde{e}') = \widetilde{l}'(\widetilde{v}')\}$. We have 

    \begin{equation}
        \sum_{\bm{\widetilde{o}}_{\widetilde{v}'} \in \widetilde{O}'_{\widetilde{v}'}} p'(\bm{\widetilde{o}}_{\widetilde{v}'} | \widetilde{v}') = 
        \begin{cases}
            \sum_{\bm{\widetilde{o}}_{\widetilde{v}} \in \widetilde{O}_{\widetilde{v}}} p(\bm{\widetilde{o}}_{\widetilde{v}} | \widetilde{v}) = 1, & \text{when } \widetilde{v}' = \phi(\widetilde{v}) \text{ for some } \widetilde{v} \in \widetilde{V};\\
            \sum_{{\widetilde{o}}_{\widetilde{v}}(\widetilde{e}) \in \{+1,-1\}} p({\widetilde{o}}_{\widetilde{v}}(\widetilde{e}) | \widetilde{v}) = 1, & \text{when } \widetilde{v}' \in \text{PATH}_{\widetilde{u}, \widetilde{v}} \text{ for some } \widetilde{u}, \widetilde{v}\in \widetilde{V}' \text{ and } \widetilde{v}' \neq \phi(\widetilde{u}), \widetilde{v}' \neq \phi(\widetilde{v}), \widetilde{e} = (\widetilde{u}, \widetilde{v});\\
            p'(+1, +1, \ldots, +1 | \widetilde{v}') = 1, & \text{otherwise.}
        \end{cases}
    \end{equation}

    For the non-disturbing condition, we have:
\begin{equation}
    \begin{split}
        p'(\widetilde{o}_{\widetilde{v}'}(\widetilde{e}') = +1 | \widetilde{v}') = 
        \begin{cases}
            p(\widetilde{o}_{\widetilde{v}}(\widetilde{e}) = +1 | \widetilde{v}) & \text{when both } \widetilde{v}', \widetilde{e}' \text{ are in } \text{PATH}_{\widetilde{u}, \widetilde{v}} \text{ for some } \widetilde{u} \in \widetilde{V}, \text{ and } \widetilde{e} = (\widetilde{u}, \widetilde{v});\\
            1 & \text{otherwise}.
        \end{cases} \\
        p'(\widetilde{o}_{\widetilde{v}'}(\widetilde{e}') = -1 | \widetilde{v}') = 
        \begin{cases}
            p(\widetilde{o}_{\widetilde{v}}(\widetilde{e}) = -1 | \widetilde{v}) & \text{when both } \widetilde{v}', \widetilde{e}' \text{ are in } \text{PATH}_{\widetilde{u}, \widetilde{v}} \text{ for some } \widetilde{u} \in \widetilde{V}, \text{ and } \widetilde{e} = (\widetilde{u}, \widetilde{v});\\
            0 & \text{otherwise}.       
        \end{cases}
    \end{split}
    \end{equation}
\end{proof}

Now we proceed to prove the main claim of this subsection.

\begin{proposition}
    Magic arrangement-based randomness generation protocols are not secure against hybrid classical adversaries who can prepare non-disturbing behaviors for the device.
\end{proposition}

\begin{proof}
    To prove this proposition, we show that for both the magic square and the magic pentagram, for any hyperedge $e^*\in E$, there always exists a non-disturbing realization that assigns deterministic outcomes to the vertices belonging to $e^*$. Specifically, for any hyperedge $e^* \in E$, there exists a no-signaling realization $\{p(\bm{o}_e|e)\}$ such that $ p({o}_{e^*} (v) | e^*) \in \{ 0, 1 \}$ for all $ v \in e^* $. Thus, a hybrid classical adversary capable of preparing arbitrary non-disturbing behaviors and programming them into the device can perfectly predict the outcomes of variables in $e^*$.

    In the preceding part of this subsction, we provided the non-disturbing realizations for the signed arrangements shown in Fig.~\ref{arrangement_graph}. Both graphs, $K_{2,2}$ and $K_3$, are isomorphic to their respective intersection graphs. It is evident that $K_{2,2}$ (see Fig.~\ref{k22}) is a topological minor of the graph $K_{3,3}$ (see Fig.~\ref{MSquare}), which represents the intersection graph of the magic square. For any specific hyperedge $e_{s}^*$ of the magic square, there exists an injective map $\phi$ (defined on the intersection graphs) from $K_{2,2}$ to $K_{3,3}$, as explained in Definition~\ref{def_topo_minor}, such that the preimage of $e_{s}^*$ is empty. Due to the non-disturbing realization construction presented in Proposition~\ref{app_prop2}, there exists a non-disturbing realization of the magic square in which the outcomes of variables in $e_{s}^*$ are always $+1$, i.e., deterministic.

    A similar phenomenon occurs for the magic pentagram. This occurs because $K_3$ (see Fig.~\ref{k3}) is a topological minor of the graph $K_5$ (see Fig.~\ref{Mpentagram}), which is the intersection graph of the magic pentagram. For any hyperedge $e_{p}^*$ of the magic pentagram, there exists an injective map $\phi'$ (defined on the intersection graphs) from $K_3$ to $K_5$, as explained in Definition~\ref{def_topo_minor}, such that the preimage of $e_{p}^*$ is empty. Therefore, there exists a non-disturbing realization of the magic pentagram in which the outcomes of variables in $e_{p}^*$ are deterministic.

    Based on the statement above and the fact~\cite{arkhipov2012extending} that the intersection graph of any magic arrangement always contains $K_{3,3}$ (the intersection graph of the magic square) or $K_5$ (the intersection graph of the magic pentagram) as a topological minor, along with the construction of the non-disturbing realizations described in Propositions~\ref{app_prop1} and \ref{app_prop2}, we can conclude that such attacks hold for any magic arrangements.
\end{proof}